\documentclass{ieeeaccess}

\usepackage{amsfonts,amsthm,braket,amsmath,amssymb,algorithm2e}
\usepackage{graphicx}
\usepackage{textcomp}
\usepackage[colorlinks, citecolor=blue, linkcolor=blue,urlcolor=blue]{hyperref}

\newtheorem{theorem}{Theorem}[]

\newtheorem{lemma}[theorem]{Lemma}

\newcommand\bket[1]{|\mathfrak{#1}\rangle}
\newcommand\bbra[1]{\langle\mathfrak{#1}|}

\newcommand \kb[1]{ \ket{#1}\bra{#1} }
\newcommand \bkb[1]{ \bket{#1}\bbra{#1} }
\newcommand \bk[1]{ \braket{#1|#1} }

\newcommand \pdouble { p_{\texttt{double}} }
\newcommand \pctrl[1] { p_{\texttt{CTRL:} #1} }
\newcommand \pcreate[1] { p_{\texttt{create:} #1} }

\newcommand\sA[0]{\mathrm{A}}
\newcommand\sB[0]{\mathrm{B}}
\newcommand\sE[0]{\mathrm{E}}
\newcommand\sZ[0]{\mathrm{Z}}
\newcommand\sX[0]{\mathrm{X}}
\newcommand\sF[0]{\mathrm{F}}
\newcommand\sR[0]{\mathrm{R}}
\newcommand\sx[0]{\mathrm{x}}
\newcommand\pr[0]{\mathrm{Pr}}

\newlength{\xfigwd}
\RestyleAlgo{boxruled}
\LinesNumbered

\begin{document}

\history{}
\doi{10.1109/TQE.2023.3261262}

\title{Security Proof Against Collective Attacks
for an Experimentally Feasible Semiquantum Key Distribution Protocol}

\author{\uppercase{Walter O. Krawec}\authorrefmark{1},
\uppercase{Rotem Liss}\authorrefmark{2,3}, and \uppercase{Tal Mor}\authorrefmark{2}}
\address[1]{Department of Computer Science and Engineering,
  University of Connecticut, Storrs CT USA (email: walter.krawec@uconn.edu)}
\address[2]{Computer Science Department, Technion, Haifa 3200003,
Israel}
\address[3]{D\'epartement IRO, Universit\'e de Montr\'eal,
Montr\'eal, Qu\'ebec H3C 3J7, Canada}
\tfootnote{The work of Walter O.\ Krawec was supported in part
by the National Science Foundation under Grant~1812070.
The work of Rotem Liss and Tal Mor was supported in part
by the Research and Technology Unit of the Israeli Ministry of Defense
and in part by the Helen Diller Quantum Center at the Technion.
The work of Rotem Liss was supported in part
by the Canada Research Chair Program.}

\begin{abstract}
Semiquantum key distribution (SQKD) allows two parties (Alice and Bob)
to create a shared secret key, even if one of these parties
(say, Alice) is classical. However, most SQKD protocols suffer
from severe practical security problems
when implemented using photons.
The recently developed ``Mirror protocol'' \cite{mirror17}
is an experimentally feasible SQKD protocol overcoming those drawbacks.
The Mirror protocol was proven robust
(namely, it was proven secure against
a limited class of attacks including all noiseless attacks),
but its security in case some noise is allowed
(natural or due to eavesdropping)
has not been proved yet.
Here we prove security of the Mirror protocol
against a wide class of quantum attacks (the ``collective attacks''),
and we evaluate the allowed noise threshold and the resulting key rate.
\end{abstract}

\begin{keywords}
Quantum information, quantum key distribution (QKD), security
\end{keywords}

\titlepgskip=-15pt

\maketitle

\section{Introduction}\label{sec_intro}
Quantum key distribution (QKD) protocols~\cite{bb84} make it possible
for two parties, Alice and Bob, to generate a secret shared
key. This key is information-theoretically secure against any possible
attack that can be applied by an all-powerful adversary Eve
limited only by the laws of physics.

Semiquantum key distribution (SQKD)~\cite{cbob07} allows Alice and Bob
to achieve the same goal even if one of them is \textit{classical} in nature.
For example, an SQKD protocol can allow a classical Alice and
a quantum Bob to generate a secret shared key, where Alice is only
allowed to perform operations
in the computational basis $\{\bket{0}, \bket{1}\}$
and reflect qubits that go through her laboratory untouched.
Previously suggested SQKD protocols include
``QKD with Classical Bob''~\cite{cbob07,sqkd09},
``QKD with Classical Alice''~\cite{calice09,calice09comment},
and many others
(e.g.,~\cite{LC2008,SDL2013,YYLH2014,mediated15,ZQZM15,zhang2020single,rong2021mediated,rong2020two,rong2020semi}).

SQKD protocols use the notion of ``classical operations''
performed by a ``classical party''.
However, in the 15 years since the publication of the original paper
introducing SQKD protocols~\cite{cbob07}, we noticed that the term
``classical party'' sometimes causes confusion:
in other hybrid quantum-classical protocols
described in the literature (e.g.,~\cite{mahadev18,reichardt2013classical}),
the term classical operations is kept only
to operations performed on \textit{classical} bits,
and it is implicitly or explicitly assumed that all classical parties
have no access to quantum states (e.g., qubits)
and cannot perform any operation on them.
On the other hand, classical parties in SQKD protocols
\textit{can} perform limited operations on quantum states.

To avoid this confusion, we introduce here the notion of
\textit{CloQ}---\textit{Classical Operations on Quantum Data}.
CloQ protocols involve at least one classical party
(or CloQ party) who is restricted to using
the four classical operations~\ref{cloq_op1}--\ref{cloq_op4} described below
for interacting with a quantum channel.
CloQ protocols have been shown to exhibit highly interesting
theoretical properties; currently, their most well understood application
is SQKD (see~\cite{sqc20} for a recent review),
but CloQ protocols have also been devised
to solve other cryptographic problems, including
secret sharing~\cite{secret-share-1,secret-share-2,secret-share-3},
secure direct
communication~\cite{semi-sdc-1,semi-sdc-2,semi-sdc-3,semi-sdc-4},
identity verification~\cite{semi-ident-1,semi-ident-2},
and private state comparison~\cite{semi-compare}.
CloQ protocols may even be devised in the future for quantum
verification by defining a CloQ variant of QPIP (quantum prover
interactive proofs)~\cite{aharonov2017interactive},
which could allow a CloQ party to verify
quantum computations performed by a fully quantum center (or prover).
Possible generalizations of this idea include verification protocols
for a CloQ verifier and a computationally unbounded prover
(a known concept in complexity theory),
as well as blind verification protocols
where the quantum prover is oblivious to the computations it performs
at the CloQ verifier's request.

The classical party in a CloQ protocol is restricted
to limited classical operations but is capable of performing
these operations on a quantum communication channel.
Such protocols rely on a two-way quantum channel,
which makes security analyses difficult
(similarly to other two-way QKD protocols;
see, e.g.,~\cite{plug_play97,ping_pong02,LM05,renner-twoway13}),
especially in practical and experimental settings allowing
a quantum state to travel from one party to the classical party
and back to the original sender. The classical party is restricted
to using the following classical operations
(see, e.g.,~\cite{cbob07,sqkd09}):
\begin{enumerate}
\item \label{cloq_op1}Preparing a qubit in one of the computational basis
states: $\bket{0}$ or $\bket{1}$.
\item \label{cloq_op2}Measuring a qubit in the computational basis
$\{\bket{0}, \bket{1}\}$.
\item \label{cloq_op3}Ignoring the qubit, letting it pass through their lab
back to the sender undisturbed.
\item \label{cloq_op4}Permuting incoming qubits and returning them to the sender
in a new order, but otherwise undisturbed.
\end{enumerate}

CloQ protocols, and in particular SQKD protocols,
are fascinating from a theoretical point of view
because they attempt to find out ``how quantum'' a protocol must be
to gain an advantage over a classical protocol:
for example, it is impossible to perform secure key distribution
using \textit{only} classical communication
(unless we make computational assumptions),
but SQKD protocols show that one classical party
and one quantum party \textit{can} achieve information-theoretically
secure key distribution.

While the importance of SQKD protocols is clear from a theoretical standpoint,
their practical importance is more subtle.
Since the practical implementation of fully-quantum QKD (e.g.,
BB84) is a well studied problem with numerous high-speed implementations, the
reader may rightly wonder at the practical importance of studying SQKD protocols.
However, there are several advantages to this
study from a practical perspective. First, semi-quantum communication is
a practical technology, as some experimental proofs of concept have been
demonstrated~\cite{mirror_gurevich13,mirror_tamari14,massa2022experimental}.
Second, while these experimental proofs of concept
required hardware similar to their fully-quantum counterparts,
the ability to perform CloQ operations may become cheaper
as technology advances,
so it is important to study alternative implementation methodologies now.
Third, several semi-quantum protocols rely on imbalanced user
capabilities---for example, the fully quantum user can invest
in higher quality equipment, while the classical user can rely
on cheaper devices (e.g., measurement devices with lower efficiency),
leading to interesting use-case scenarios. Fourth, the security proof
methodologies developed for practical SQKD protocols can be translated to other
QKD protocols with potential new insights and countermeasure strategies;
for example, proof techniques developed for practical SQKD
can demonstrate how to compensate for imperfect or
imbalanced hardware capabilities or partial device failures.
Last but not least, if one wants to hide from some of the users the fact
that quantum cryptography is used, the true description of the classical
operations~\ref{cloq_op1}--\ref{cloq_op4} can indeed
hide any hint from such an oblivious party;
after all, also when using classical data one can
either check if the bit is $0/1$ or choose to avoid checking it.
Taken together, not only is the study of SQKD protocols
(and CloQ protocols in general) important from a theoretical standpoint,
but it can also have highly interesting practical implications.

However, while the capabilities of SQKD protocols
in the ideal (perfect-qubit) scenario are now fairly well understood,
and while in principle such protocols could allow simpler devices,
the security and performance of SQKD protocols
under \textit{practical} attacks are yet to be verified.
In fact, as pointed out by~\cite{cbob07comment,cbob07comment_reply},
many existing SQKD protocols are experimentally infeasible:
it is not known how to implement them in a secure way.
Specifically, many SQKD protocols use the SIFT classical operation,
which requires the classical user to first \textit{measure}
the incoming quantum state
in the computational basis $\{\bket{0}, \bket{1}\}$
and then \textit{resend} the measured state back to the quantum user;
the experimental implementations of this operation
are vulnerable to some ``tagging attacks'' described
by~\cite{cbob07comment,cbob07comment_reply,mirror17}.
For solving this problem, an experimentally feasible SQKD protocol
named the ``Mirror protocol'' was introduced by~\cite{mirror17};
see also~\cite{mirror_attack18}, which analyzed a simplified variant
and attacks on it.

Most SQKD protocols have been proven robust:
namely~\cite{cbob07}, if Eve obtains some secret information,
she must cause some errors that may be noticed by Alice and Bob;
equivalently, a protocol is ``robust'' if any attack that induces
no errors, must give Eve no information.
In particular, the Mirror protocol
was proven robust by~\cite{mirror17}.
Proving robustness is a step towards proving security;
proving full security of SQKD protocols is difficult
because these protocols are usually two-way: for example, Bob sends
a quantum state to Alice, and Alice performs a specific classical
operation and sends the resulting quantum state back to Bob.
A few SQKD protocols also have a security
analysis~\cite{cbob_secur15,sqkd_secur16,calice_secur18,sqkd_secur18}
which is usually applicable to an ideal qubit-based description,
but not to the more realistic photon-based description.
So far, the Mirror protocol has not been proven secure.

In this paper we prove security of the Mirror protocol
against collective attacks.
The class of the collective attacks~\cite{BM97a,BM97b,BBBGM02}
is an important and powerful subclass of possible attacks;
the class of the general attacks
(also known as the joint attacks;
see, e.g.,~\cite{mayers01,SP00,BBBMR06,RGK05})
includes all theoretical attacks allowed by quantum physics.
Security against collective attacks is conjectured
(and, in some security notions, proved~\cite{renner_thesis08,CKR09,geat2022})
to imply security against general attacks.
However, some existing security proofs of SQKD protocols
against general attacks may in fact be limited to collective attacks,
because they use de Finetti's theorem and similar techniques
(see~\cite{renner_thesis08,CKR09}) that can directly be applied
only to entanglement-based protocols~\footnote{Applying
de Finetti's theorem and similar techniques
to prepare-and-measure protocols (including SQKD protocols)
is usually easy for one-way QKD protocols,
but it does not necessarily work for two-way protocols.}. In particular,
to use these techniques, one usually requires some reduction
from the two-way protocol to an entanglement-based protocol.
Such reduction techniques are known for certain classes
of two-way protocols \cite{renner-twoway13,guskind2022mediated},
but it is not known how to perform these reductions for all two-way protocols.  In particular, the method of \cite{renner-twoway13} only applies if the protocol exhibits a certain symmetry property which no semi-quantum protocol can have, while the method of \cite{guskind2022mediated} is currently only applicable to mediated semi-quantum protocols in the ideal qubit scenario.
In particular, these previous techniques
do not apply to the Mirror protocol we consider in this work.
Therefore, in this paper we restrict our analysis
to collective attacks.

This paper proves security of the Mirror protocol
under a large class of collective attacks,
which include the ability of Eve to inject \textit{multiple photons}
into the classical user's lab, but not into the quantum user's lab
(attacks of the later kind are left for future analysis,
but we briefly discuss them in the beginning of
Section~\ref{sec_proof}).
In addition, we limit our analysis
to two-mode quantum communication,
leaving more complicated attacks for future research.
We assume Alice's and Bob's devices precisely implement
the needed operations
(most notably, Alice's classical operations described
in Eqs.~\eqref{eq:I-def}--\eqref{eq:S-def}),
and without loss of generality, we assume an all-powerful Eve
controlling all errors and losses in the quantum channel.

We derive an information-theoretic proof of security
against these attacks and simulate the performance of the protocol
in a variety of realistic scenarios, including lossy quantum channels,
compared to the BB84 protocol. Ultimately, our paper shows
that SQKD protocols hold the potential to be secure and feasible
in practice, and not just ``secure in ideal conditions''.
The methods and techniques we present in this work
may also be applicable to security proofs of other SQKD protocols
or even other two-way QKD protocols where users are limited
in some manner in their quantum capabilities.

\section{The Mirror Protocol}\label{sec_mirror}
This section is partially based on~\cite{mirror_attack18}.

For describing the Mirror protocol,
we assume a photonic implementation consisting of two modes:
the mode of the qubit state $\bket{0}$
and the mode of the qubit state $\bket{1}$ (below we call them
``the $\bket{0}$ mode'' and ``the $\bket{1}$ mode'', respectively).
For example, the $\bket{0}$ mode and the $\bket{1}$ mode
can represent two different polarizations or two different time bins.
As elaborated in~\cite{mirror17}, the Mirror protocol can intuitively
be described in terms of photon pulses
that correspond to two distinct time bins,
which means that the classical party (Alice)
can only perform operations on the two distinct time bins
(corresponding to the computational basis $\{\bket{0}, \bket{1}\}$)
and not on their superpositions
(corresponding, for example,
to the Hadamard basis $\{\bket{+}, \bket{-}\}$).

\begin{figure}
\includegraphics[width=0.5\textwidth]{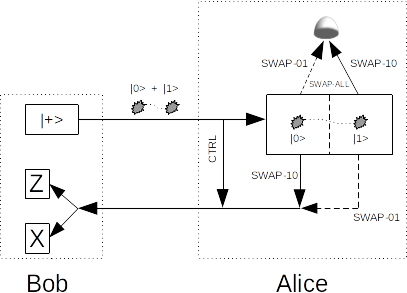}
\caption{A schematic diagram of the Mirror protocol
described in Section~\ref{sec_mirror}.}\label{fig:diagram}
\end{figure}

\subsection{The Single-Photon Case}\label{subsec_single}
We use the Fock space notations: if there is exactly one photon,
the Fock state $\ket{0,1}$
represents one photon in the $\bket{0}$ mode,
and the Fock state $\ket{1,0}$
represents one photon in the $\bket{1}$ mode
(and, thus, our Hilbert space is the qubit space
$\mathrm{Span}\{\ket{0,1}, \ket{1,0}\}$).
We can extend the qubit space to a $3$-dimensional Hilbert space
by adding the Fock ``vacuum state'' $\ket{0,0}$,
which represents an absence of photons.
Similarly, in the Hadamard basis,
we define the Fock state $\ket{0,1}_\sx \triangleq
\frac{\ket{0,1} + \ket{1,0}}{\sqrt{2}}$
(equivalent to the $\bket{+}$ state)
and the Fock state $\ket{1,0}_\sx \triangleq
\frac{\ket{0,1} - \ket{1,0}}{\sqrt{2}}$
(equivalent to the $\bket{-}$ state).

In the Mirror protocol (regardless of the specific implementation),
in each round,
Bob sends to Alice the initial state $\bket{+}_\sB$,
which is equivalent to $\ket{0,1}_{\sx,\sB} \triangleq
\frac{\ket{0,1}_\sB + \ket{1,0}_\sB}{\sqrt{2}}$.
Then, Alice prepares an ancillary state in the initial vacuum state
$\ket{0,0}_{\sA_{\mathrm{anc}}}$ and
chooses \textit{at random} one of the following
four classical operations (defined on any Fock state
she may possibly get, due to Eve's single-photon attacks
possible in this case):
\begin{itemize}
\item $\mathbf{I}$ (CTRL)  Reflect all photons towards Bob,
without measuring any photon. The mathematical description is:
\begin{equation}
I \ket{0,0}_{\sA_{\mathrm{anc}}} \ket{m_1,m_0}_\sB
= \ket{0,0}_{\sA_{\mathrm{anc}}} \ket{m_1,m_0}_\sB.\label{eq:I-def}
\end{equation}
\item $\mathbf{S_1}$ (SWAP-10) Reflect all photons in the
$\bket{0}$ mode towards Bob, and measure all photons in the
$\bket{1}$ mode. The mathematical description is:
\begin{equation}
S_1 \ket{0,0}_{\sA_{\mathrm{anc}}} \ket{m_1,m_0}_\sB
= \ket{{m_1,0}}_{\sA_{\mathrm{anc}}} \ket{0,m_0}_\sB.\label{eq:S1-def}
\end{equation}
\item $\mathbf{S_0}$ (SWAP-01) Reflect all photons in the
$\bket{1}$ mode towards Bob, and measure all photons in the
$\bket{0}$ mode. The mathematical description is:
\begin{equation}
S_0 \ket{0,0}_{\sA_{\mathrm{anc}}} \ket{m_1,m_0}_\sB
= \ket{0,m_0}_{\sA_{\mathrm{anc}}} \ket{m_1,0}_\sB.\label{eq:S0-def}
\end{equation}
\item $\mathbf{S}$ (SWAP-ALL) Measure all photons, without
reflecting any photon towards Bob. The mathematical description is:
\begin{equation}
S \ket{0,0}_{\sA_{\mathrm{anc}}} \ket{m_1,m_0}_\sB
= \ket{m_1,m_0}_{\sA_{\mathrm{anc}}} \ket{0,0}_\sB.\label{eq:S-def}
\end{equation}
\end{itemize}

We note that in the above mathematical description,
Alice measures her ancillary state $\ket{\cdot}_{\sA_{\mathrm{anc}}}$
in the computational basis $\{\bket{0}, \bket{1}\}$
and sends back to Bob the $\ket{\cdot}_\sB$ state.

The states sent from Alice to Bob (without any error, loss,
or eavesdropping) and their interpretations,
depending on Alice's random choice of a classical operation
and on whether Alice detected a photon or not,
are detailed in Table~\ref{table_expected_results}.

\begin{table*}[htpb]
\caption{The state sent from Alice to Bob in the Mirror protocol
without errors or losses, and its interpretation,
depending on Alice's random choice of a classical operation
and on whether Alice detected a photon or not.}
\label{table_expected_results}
\centering
\begin{tabular}{ccccccccc}
\hline\hline
\textbf{Alice's Operation} & &
\textbf{Did Alice Detect a Photon?} & &
\textbf{State Sent to Bob} & &
\textbf{Round Type} & &
\textbf{Raw Key Bit} \\
\hline
CTRL & & no (happens with certainty) & & $\ket{0,1}_{\sx,\sB}$
& & ``test'' & & none \\
\hline
SWAP-10 & & no (happens with probability $\frac{1}{2}$) & &
$\ket{0,1}_\sB$ & & ``raw key'' & & 0 \\
SWAP-10 & & yes (happens with probability $\frac{1}{2}$) & &
$\ket{0,0}_\sB$ & & ``raw key'' & & none \\
\hline
SWAP-01 & & no (happens with probability $\frac{1}{2}$) & &
$\ket{1,0}_\sB$ & & ``raw key'' & & 1 \\
SWAP-01 & & yes (happens with probability $\frac{1}{2}$) & &
$\ket{0,0}_\sB$ & & ``raw key'' & & none \\
\hline
SWAP-ALL & & yes (happens with certainty) & &
$\ket{0,0}_\sB$ & & ``SWAP-ALL'' & & none \\
\hline\hline
\end{tabular}
\end{table*}

\subsection{The Multi-Photon Case}\label{subsec_multi}
Most generally, we need to describe Alice's operation
on a general state, because Eve can attack the state sent
from Bob to Alice. The Fock state $\ket{m_1, m_0}$ represents
$m_1$ indistinguishable photons in the $\bket{1}$ mode
and $m_0$ indistinguishable photons in the $\bket{0}$ mode.
More details about the Fock space notations are given
in~\cite{mirror17}; using these mathematical notations is vital
for describing and analyzing all practical attacks
on a QKD protocol (see~\cite{BLMS00} for details and examples).

The mathematical description of the Mirror protocol
in this multi-photon case remains \textit{identical}
to its description in Subsection~\ref{subsec_single}.
However, in this case, Alice's classical operations are defined
on \textit{any} general Fock state, because Eve's attack can include
any multi-photon pulse.

\subsection{Bob's Final Measurements and Classical Post-Processing}
In both cases described
in Subsections~\ref{subsec_single} and~\ref{subsec_multi},
Bob finally measures the incoming state in a random basis
(either the computational basis $\{\bket{0}, \bket{1}\}$ or
the Hadamard basis $\{\bket{+}, \bket{-}\}$).
We assume here, as is true in most experimental setups, that 
Alice and Bob use detectors and not counters:
namely, their detectors cannot \textit{count}
the number of incoming photons. Therefore, when a detector clicks,
Alice and Bob cannot know whether it detected a single-photon pulse
(a single photon in its measured mode)
or a multi-photon pulse (more than one photon in its measured mode).

After completing all rounds, Alice and Bob
perform \textit{classical post-processing}:
Alice sends over the classical channel
her operation choices (CTRL, SWAP-$x$, or SWAP-ALL;
she keeps $x \in \{01,10\}$ in secret);
Bob sends over the classical channel his basis choices;
and both of them reveal all rounds where they got a loss,
and all measurement results each of them got in all testing rounds
(CTRL, SWAP-ALL, and a random subset of the SWAP-$x$ rounds,
for which Alice also reveals her values of $x \in \{01,10\}$)
and in all mismatched rounds (such as rounds in which
Alice used SWAP-10 and Bob used the Hadamard basis).

In the non-testing rounds,
as detailed in Table~\ref{table_expected_results},
Alice and Bob share the raw key bit 0 if Alice uses
SWAP-10 and detects no photon while Bob measures in the
computational basis and detects a photon (or photons)
in the $\bket{0}$ mode;
similarly, they share the raw key bit 1 if Alice uses
SWAP-01 and detects no photon while Bob measures in the
computational basis and detects a photon (or photons)
in the $\bket{1}$ mode.

Now, Alice and Bob have enough information for computing
all the probabilities they need for finding the key rate
(that are detailed later, in Table~\ref{table_simple_probs}),
so they compute all these probabilities and
deduce the final key rate according to the algorithm in
Subsection~\ref{subsec_algo}.
If the final key rate is negative, they abort the protocol;
otherwise, they perform error correction and privacy amplification
in the standard way for QKD protocols.
At the end of the protocol, Alice and Bob hold an identical final key
that is completely secure against any eavesdropper.

A full description of the Mirror protocol and
a proof of its robustness are both available in~\cite{mirror17}.
An illustration of the Mirror protocol is available as
Fig.~\ref{fig:diagram}.

\section{Security Proof of the Mirror Protocol
Against Collective Attacks}\label{sec_proof}

\begin{table*}[htpb]
\caption{All the probabilities Alice and Bob need to estimate
in order to compute the key rate in Theorem~\ref{thm:main-thm}.}
\label{table_simple_probs}
\centering
\begin{tabular}{lll}
\hline\hline
\textbf{Notation} & \textbf{Definition} & \textbf{Round Type This Occurs}\\
\hline
$\bk{E_0}_\sE$  & Probability that Alice and Bob get raw key bits
$0,0$, respectively & ``raw key''\\
$\bk{E_1}_\sE$ &  Probability that Alice and Bob get raw key bits
$0,1$, respectively  & ``raw key''\\
$\bk{E_2}_\sE$ &  Probability that Alice and Bob get raw key bits
$1,0$, respectively  & ``raw key''\\
$\bk{E_3}_\sE$  & Probability that Alice and Bob get raw key bits
$1,1$, respectively  & ``raw key''\\
$M$  & Probability that both Alice and Bob get raw key bits & ``raw key''\\
$p_{0,+}$  & Probability that Alice gets raw key bit $0$,
and Bob observes $\ket{+}$ & ``raw key'' (with mismatched bases)\\
$p_{1,+}$ & Probability that Alice gets raw key bit $1$,
and Bob observes $\ket{+}$ & ``raw key'' (with mismatched bases)\\
\hline
$p_{+,+}$ & Probability that Bob observes $\ket{+}$ & ``test''\\
$\pctrl{0}$  & Probability that Bob observes $\ket{0,1}$& ``test'' (with mismatched bases)\\
$\pctrl{1}$ & Probability that Bob observes $\ket{1,0}$& ``test'' (with mismatched bases)\\
\hline
$\pdouble$ & Probability that Alice observes a ``double-click'' event
($\ket{1,1}$) & ``SWAP-ALL''\\
$\pcreate{0}$  & Probability that Alice observes $\ket{0,0}$, and Bob
observes $\ket{0,1}$ & ``SWAP-ALL''\\
$\pcreate{1}$  & Probability that Alice observes $\ket{0,0}$, and Bob
observes $\ket{1,0}$ & ``SWAP-ALL''\\
\hline\hline
\end{tabular}
\end{table*}

We now prove security of the Mirror protocol.
For our security proof, we assume that the adversary Eve
is restricted to collective attacks---namely, that Eve attacks
each round in an independent and identical manner,
but she is allowed to postpone the measurement
of her private quantum ancilla until any future point in time.
Beyond this, we will also assume in our security analysis that
Eve is allowed to inject \textit{any} signal into the forward channel
(linking quantum Bob to classical Alice); in the reverse channel,
she is free to perform any quantum unitary probe, but we will assume
that the number of photons returning to Bob is at most one.
That is, Eve is allowed to inject
multiple photons into the channel going to Alice, but on the way back,
only a single photon or no photons at all will be returned to Bob.
This assumption means that Eve may need to remove photons
on the way from Alice to Bob, if she sent multiple photons towards
Alice; in Subsection~\ref{subsec_attacks}
we explain how Eve can perform this attack.

The above assumption (that at most one photon is sent towards Bob)
is made to simplify the analysis of the return channel.
We point out that according to~\cite{mirror17},
the Mirror protocol is completely \textit{robust}
even without this assumption---namely,
it is proved robust against \textit{all} multi-photon attacks
and \textit{all} kinds of losses and dark counts.
However, full \textit{security} analysis of the multi-photon case,
including both losses and dark counts, is very difficult 
even in the simplest one-way standard QKD,
and even more so in any standard two-way QKD protocol
such as ``Plug \& Play''~\cite{plug_play97},
``Ping Pong''~\cite{ping_pong02},
and LM05~\cite{LM05} (see also~\cite{renner-twoway13}).
Furthermore, this case has not been analyzed
in security proofs of many other SQKD protocols
(e.g.,~\cite{cbob_secur15,sqkd_secur16,calice_secur18,sqkd_secur18}).
Therefore, we do not aim to solve this major issue here
in the specific case of the Mirror protocol:
extending the full security proof to this most general case
is left for future research.

Our main result in this section is a lower bound on the von Neumann entropy
$S(A|E)$ of the protocol. This allows us to determine a lower bound
on the key rate of the protocol
using the Devetak-Winter key rate equation~\cite{DW05}.
Our main key rate result is summarized in the following theorem
(which uses notations defined in Table~\ref{table_simple_probs}):

\begin{theorem}\label{thm:main-thm}
  Assuming the attack model discussed above, consider the observable statistics and their respective notations listed in Table~\ref{table_simple_probs}.  Then, the key rate of the protocol is lower-bounded by:
  \begin{eqnarray}
\texttt{rate} &\ge& \frac{\braket{E_0|E_0}_\sE + \braket{E_3|E_3}_\sE}{M}
\label{eq:SAE-HAB-bound} \\
&\times&
\left[ H_2\left(\frac{\bk{E_0}_\sE}{\bk{E_0}_\sE + \bk{E_3}_\sE}\right)
- H_2(\lambda_1) \right]\nonumber\\
&+& \frac{\braket{E_1|E_1}_\sE + \braket{E_2|E_2}_\sE}{M}
\nonumber \\
&\times&
\left[ H_2\left(\frac{\bk{E_1}_\sE}{\bk{E_1}_\sE + \bk{E_2}_\sE}\right)
  - H_2(\lambda_2) \right]\nonumber\\
&-& H(A|B),\nonumber
\end{eqnarray}
where:
\begin{eqnarray*}
\lambda_1 &\triangleq& \frac{1}{2} \\
&+& \frac{\sqrt{\left(\bk{E_0}_\sE - \bk{E_3}_\sE\right)^2
+ 4\Re^2\braket{E_0|E_3}_\sE}}
{2\left(\bk{E_0}_\sE+\bk{E_3}_\sE\right)},
\end{eqnarray*}
\begin{eqnarray*}
\lambda_2 &\triangleq& \frac{1}{2} \\
&+& \frac{\sqrt{\left(\bk{E_1}_\sE - \bk{E_2}_\sE\right)^2
+ 4\Re^2\braket{E_1|E_2}_\sE}}
{2\left(\bk{E_1}_\sE+\bk{E_2}_\sE\right)},\\
H_2(x) &\triangleq& -x \log_2(x) - (1-x) \log_2(1-x),
\end{eqnarray*}
\begin{align*}
&H(A|B) \\
&= H \left( \frac{\bk{E_0}_\sE}{M}, \frac{\bk{E_1}_\sE}{M},
\frac{\bk{E_2}_\sE}{M}, \frac{\bk{E_3}_\sE}{M} \right) \\
&- H \left( \frac{\bk{E_0}_\sE + \bk{E_2}_\sE}{M},
\frac{\bk{E_1}_\sE + \bk{E_3}_\sE}{M} \right), \\
&H(x_1, \ldots, x_k) \triangleq -\sum_{j=1}^k x_j \log_2(x_j),
\end{align*}
subject to the following constraint:
\begin{align}
&\Re\left(\braket{E_0|E_3}_\sE + \braket{E_1|E_2}_\sE\right)\nonumber\\
&\ge \frac{1}{2} p_{+,+} - p_{0,+} - p_{1,+}
- \frac{1}{4}(\pctrl{0}+\pctrl{1}) + \frac{1}{2} M \nonumber \\
&- \frac{1}{\sqrt{2}}\left(\sqrt{\pcreate{1}} + \sqrt{\pdouble}\right)
\left(\sqrt{\bk{E_0}_\sE} + \sqrt{\bk{E_2}_\sE}\right) \nonumber \\
&- \frac{1}{\sqrt{2}}\left(\sqrt{\pcreate{0}} + \sqrt{\pdouble}\right)
\left(\sqrt{\bk{E_1}_\sE} + \sqrt{\bk{E_3}_\sE}\right) \nonumber \\
&- \frac{1}{2}\left(\sqrt{\pcreate{0}} + \sqrt{\pdouble}\right)
\left(\sqrt{\pcreate{1}} + \sqrt{\pdouble}\right).
\end{align}
\end{theorem}

We prove Theorem~\ref{thm:main-thm} in several steps.
First, in Subsection~\ref{subsec_attacks}
we describe Eve's most general attacks
that are allowed under our attack model assumptions.
Following this, in Subsection~\ref{subsec_raw_key}
we present the final quantum state $\rho_{\mathrm{ABE}}$
shared by Alice, Bob, and Eve at the end of each round of the protocol,
conditioning on a raw-key bit being generated during that round.
To complete the proof, we must find a lower bound
on the conditional von Neumann entropy $S(A|E)$
corresponding to $\rho_{\mathrm{ABE}}$.
For this, in Subsections~\ref{subsec_iter}--\ref{subsec_swap_all} we show
how Alice and Bob can use observable probabilities
from all types of rounds (see Table~\ref{table_probs}) to compute
inner products and norms of quantum states appearing in $\rho_{\mathrm{ABE}}$.
Then, in Subsection~\ref{subsec_keyrate}
we use a theorem from~\cite{QKD-Tom-Krawec-Arbitrary} to compute
the von Neumann entropy of $\rho_{\mathrm{ABE}}$ as a function
of our computed inner products.
Finally, in Subsection~\ref{subsec_keyrate} we combine all results
from Subsections~\ref{subsec_iter}--\ref{subsec_swap_all} to find lower bounds
on the required inner products as functions
of the observable probabilities from Table~\ref{table_probs},
which completes the proof of Theorem~\ref{thm:main-thm}.

\subsection{Eve's Attacks}\label{subsec_attacks}
\paragraph{Eve's first attack}
We first analyze the forward-channel attack---namely, the attack
on the way from Bob to Alice. Here, we note that it is to Eve's
advantage to simply discard the signal coming from Bob
(which should be the same each round and carries no information
at this point) and inject a signal of her own, possibly consisting
of multiple photons and entangled with her private quantum ancilla.

Specifically, in each round, Bob sends to Alice
the same quantum state: $\ket{0,1}_{\sx,\sB} \triangleq
\frac{\ket{0,1}_\sB + \ket{1,0}_\sB}{\sqrt{2}}$.
At this point, Eve performs her \textit{first} attack:
she replaces Bob's original state by her own state.
Since Bob never prepares alternative initial states,
Eve dropping the signal and replacing it with one of her own
is the most general strategy she could perform
in the collective attack scenario.
Without loss of generality, Eve's state is of the form:
\begin{equation}\label{eq:init-state}
\ket{\psi_0} \triangleq \sum_{\substack{m_1\ge 0\\m_0\ge 0}}
\ket{m_1, m_0}_\sB \ket{e_{m_1,m_0}}_\sE.
\end{equation}
Then, Eve sends subsystem $\sB$ to Alice
and keeps subsystem $\sE$ as her own ancillary state.
Note that as we are dealing with a two-way quantum communication channel,
Eve has two opportunities to attack the quantum signal each round.
The above equation represents the state after her first attack;
however, following Alice's encoding operation,
Eve will have a second opportunity to attack.
Unlike many one-way protocols, we cannot reduce this to
an entanglement-based protocol whereby Eve simply prepares a state
and sends part to Alice and part to Bob: although some reductions
for two-way (S)QKD protocols to equivalent entanglement-based protocols
are known~\cite{beaudry2013security,krawec2018key},
those results cannot be applied to this mirror-based protocol
and so we cannot employ them.
Thus we must analyze Eve's attack in two stages,
which makes the analysis somewhat more complicated.

\paragraph{Eve's second attack}
Then, Alice performs her classical operation (CTRL, SWAP-10, SWAP-01,
or SWAP-ALL) and sends the resulting state back to Bob.
Now, Eve performs her \textit{second} attack, described as the unitary
operator $U_\sR$. As explained above, for the second attack
we make the simplifying assumption that
Eve always sends \textit{at most one photon}---namely, she sends
a superposition of $\ket{0,1}_\sB$, $\ket{1,0}_\sB$,
and $\ket{0,0}_\sB$ with her corresponding ancillary states
$\ket{g_{m_1,m_0}^{0,1}}_\sE$, $\ket{g_{m_1,m_0}^{1,0}}_\sE$,
and $\ket{g_{m_1,m_0}^{0,0}}_\sE$.
We emphasize that this simplifying assumption applies
only to the second attack, and \textit{not} to the first attack.

Thus, Eve's second attack is of the form:
\begin{align}
& U_\sR \ket{m_1', m_0'}_\sB \ket{e_{m_1,m_0}}_\sE\nonumber\\
&= \ket{0,1}_\sB \ket{f_{m_1',m_0',m_1,m_0}^{0,1}}_\sE
+ \ket{1,0}_\sB \ket{f_{m_1',m_0',m_1,m_0}^{1,0}}_\sE\nonumber\\
&+ \ket{0,0}_\sB \ket{f_{m_1',m_0',m_1,m_0}^{0,0}}_\sE.
\label{eq:reverse-attack-op}
\end{align}
However, in our security proof we use terms of the following
simplified notations:
\begin{align}
& U_\sR \ket{m_1, m_0}_\sB \ket{e_{m_1,m_0}}_\sE\nonumber\\
&= \ket{0,1}_\sB \ket{g_{m_1,m_0}^{0,1}}_\sE
+ \ket{1,0}_\sB \ket{g_{m_1,m_0}^{1,0}}_\sE\nonumber\\
&+ \ket{0,0}_\sB \ket{g_{m_1,m_0}^{0,0}}_\sE,
\label{eq:reverse-attack-op-special}
\end{align}
where we denote $\ket{g_{m_1,m_0}^{j,k}}_\sE \triangleq
\ket{f_{m_1,m_0,m_1,m_0}^{j,k}}_\sE$.
We note that the operation of $U_\sR$ on states
$\ket{m_1',m_0'}_\sB \ket{e_{m_1,m_0}}_\sE$ where $m_1' \ne m_1$
or $m_0' \ne m_0$ will not appear in our security proof,
because these states do not give us meaningful
statistics~\footnote{States of the form
$U_\sR \ket{0,m_0}_\sB \ket{e_{m_1,m_0}}_\sE$ and
$U_\sR \ket{m_1,0}_\sB \ket{e_{m_1,m_0}}_\sE$ may appear
in ``raw key'' rounds analyzed in Subsection~\ref{subsec_raw_key},
but we analyze only rounds which contribute to the raw key,
where Alice detects no photon---namely,
$m_1 = 0$ or $m_0 = 0$, respectively.
In addition, states of the form
$U_\sR \ket{0,0}_\sB \ket{e_{m_1,m_0}}_\sE$ may appear
in ``SWAP-ALL'' rounds analyzed in Subsection~\ref{subsec_swap_all},
but we analyze only ``double-clicks'' of Alice
(where Eve's attack $U_\sR$ is irrelevant,
although we use it algebraically to prove Lemma~\ref{lemma})
and ``creation'' events
(where Alice detects no photon, so $m_1 = m_0 = 0$).} and
thus do not contribute to the probabilities in Table~\ref{table_probs}.
We also note that since Eve is all-powerful, she will have no trouble
performing any unitary operation, even if it includes
a complicated operation for reducing the number of photons.

In both attacks, subsystem $\sB$ is sent to a legitimate user,
while subsystem $\sE$ is kept as Eve's ancilla.

\subsection{Analyzing all Types of Rounds}\label{subsec_iter}
In Table~\ref{table_rounds} we classify all rounds
into six types, that Alice and Bob need to analyze.
The rounds are classified according to
Alice's random choice of a classical operation
and Bob's random choice of a measurement basis.
\begin{table}[htpb]
\caption{All types of rounds, according to 
Alice's random choice of a classical operation
[CTRL, SWAP-$x$ ($x \in \{01,10\}$), or SWAP-ALL]
and Bob's random choice of a measurement basis
(computational or Hadamard).}
\label{table_rounds}
\centering
\begin{tabular}{ccc}
\hline\hline
\textbf{Round Type} &
\textbf{Alice's Operation} &
\textbf{Bob's Basis} \\
\hline
``raw key'' & SWAP-$x$ & computational \\
mismatched ``raw key'' & SWAP-$x$ & Hadamard \\
\hline
``test'' & CTRL & Hadamard \\
mismatched ``test'' & CTRL & computational \\
\hline
``SWAP-ALL'' & SWAP-ALL & computational \\
mismatched ``SWAP-ALL'' & SWAP-ALL & Hadamard \\
\hline\hline
\end{tabular}
\end{table}

Notice the use of basis-mismatched rounds.
Technically, we could have used only the ``standard'' (basis-matching)
rounds for completing the security proof,
by using the Cauchy-Schwarz inequality for finding worst-case bounds.
However, using the technique of analyzing ``mismatched
measurements''~\cite{QKD-Tom-First,QKD-Tom-KeyRateIncrease},
we can derive a significantly improved formula for the final key rate.

Alice and Bob have to find relevant statistics for each type of round
and compute all probabilities listed in Table~\ref{table_probs}.
In Subsections~\ref{subsec_raw_key}--\ref{subsec_swap_all}
we relate these probabilities
to the quantum states appearing in our security proof,
and in Subsection~\ref{subsec_keyrate}
we derive the resulting final key rate formula.

\begin{table*}[htpb]
\caption{All the probabilities Alice and Bob need to compute,
and the formulas relating them to quantum states in our security proof.
All formulas are proved
in Subsections~\ref{subsec_raw_key}--\ref{subsec_swap_all}.}
\label{table_probs}
\centering
\begin{tabular}{llll}
\hline\hline
\textbf{Probability} & \textbf{Round} & \textbf{Definition} &
\textbf{Formula} \\
\hline
$\bk{E_0}_\sE$ & ``raw key'' & Alice and Bob get raw key bits
$0,0$, respectively & \\
$\bk{E_1}_\sE$ & ``raw key'' & Alice and Bob get raw key bits
$0,1$, respectively & \\
$\bk{E_2}_\sE$ & ``raw key'' & Alice and Bob get raw key bits
$1,0$, respectively & \\
$\bk{E_3}_\sE$ & ``raw key'' & Alice and Bob get raw key bits
$1,1$, respectively & \\
\hline
$M$ & ``raw key'' & both Alice and Bob get raw key bits &
$= \sum_{i=0}^3 \bk{E_i}_\sE$ \\
\hline
$p_{0,+}$ & mismatched & Alice gets raw key bit $0$,
and Bob observes $\ket{+}$ & $2\Re\braket{E_0|E_1}_\sE = 2 p_{0,+}$ \\
& ``raw key'' & & $- \left(\bk{E_0}_\sE + \bk{E_1}_\sE\right)$ \\
$p_{1,+}$ & mismatched & Alice gets raw key bit $1$,
and Bob observes $\ket{+}$ & $2\Re\braket{E_2|E_3}_\sE = 2 p_{1,+}$ \\
& ``raw key'' & & $- \left(\bk{E_2}_\sE + \bk{E_3}_\sE\right)$ \\
\hline
$p_{+,+}$ & ``test'' & Bob observes $\ket{+}$ & $= \left|\sum_{i=0}^3
\ket{E_i}_\sE - \sum_{j=0}^1
\left(\ket{g_j}_\sE - \ket{h_j}_\sE\right)\right|^2$
\\
\hline
$\pctrl{0}$ & mismatched & Bob observes $\ket{0,1}$ &
$= 2 \left|\ket{E_0}_\sE + \ket{E_2}_\sE
- \ket{g_0}_\sE + \ket{h_0}_\sE\right|^2$ \\
& ``test'' & & \\
$\pctrl{1}$ & mismatched & Bob observes $\ket{1,0}$ &
$= 2 \left|\ket{E_1}_\sE + \ket{E_3}_\sE
- \ket{g_1}_\sE + \ket{h_1}_\sE\right|^2$ \\
& ``test'' & & \\
\hline
$\pdouble$ & ``SWAP-ALL'' & Alice observes a ``double-click'' event
($\ket{1,1}$) &
$\bk{h_0}_\sE + \bk{h_1}_\sE \le \frac{1}{2}\pdouble$ \\
\hline
$\pcreate{0}$ & ``SWAP-ALL'' & Alice observes $\ket{0,0}$, and Bob
observes $\ket{0,1}$ & $= 2\bk{g_0}_\sE$ \\
$\pcreate{1}$ & ``SWAP-ALL'' & Alice observes $\ket{0,0}$, and Bob
observes $\ket{1,0}$ & $= 2\bk{g_1}_\sE$ \\
\hline\hline
\end{tabular}
\end{table*}

In \textit{all} types of rounds, Bob begins by sending
$\ket{0,1}_{\sx,\sB} \triangleq
\frac{\ket{0,1}_\sB + \ket{1,0}_\sB}{\sqrt{2}}$,
which Eve immediately replaces by her own state
$\ket{\psi_0} \triangleq \sum_{\substack{m_1\ge 0\\m_0\ge 0}}
\ket{m_1, m_0}_\sB \ket{e_{m_1,m_0}}_\sE$
(see Eq.~\eqref{eq:init-state}).
Then, Alice chooses her classical operation, as detailed below.

\subsection{``Raw Key'' Rounds:
Alice Chooses the SWAP-$x$ Operation}\label{subsec_raw_key}
In ``raw key'' rounds, Alice chooses either SWAP-10 or SWAP-01
(each with probability $\frac{1}{2}$),
that are defined in Eqs.~\eqref{eq:S1-def}--\eqref{eq:S0-def}.
Then, the non-normalized state of the joint system,
conditioning on Alice detecting \textit{no photon}~\footnote{Notice
that according to Table~\ref{table_expected_results},
raw key bits are shared by Alice and Bob
only in ``raw key'' rounds where Alice detects \textit{no photon}
and Bob \textit{does} detect a photon.}, is:
\begin{align}
&\rho_{\mathrm{ABE}}^{(\text{after Alice})} \nonumber \\
&= \frac{1}{2} \bkb{0}_\sA \otimes P \left( \sum_{m_0 \ge 0}
\ket{0, m_0}_\sB \ket{e_{0,m_0}}_\sE \right) \nonumber \\
&+ \frac{1}{2} \bkb{1}_\sA \otimes P \left( \sum_{m_1 \ge 0}
\ket{m_1, 0}_\sB \ket{e_{m_1,0}}_\sE \right),
\end{align}
where we define:
\begin{equation}
P(\ket{\psi}) \triangleq \ket{\psi} \bra{\psi}.
\end{equation}
We note that $\bket{0}_\sA$ and $\bket{1}_\sA$ denote the raw key bit
of Alice: Alice deduces it from her own choice
of SWAP-10 (which corresponds to $\bket{0}_\sA$)
or SWAP-01 (which corresponds to $\bket{1}_\sA$),
as explained in Table~\ref{table_expected_results}.

After Eve's second attack
(namely, after Eve applies the $U_\sR$ operator
defined in Eq.~\eqref{eq:reverse-attack-op-special}),
the joint non-normalized state becomes:
\begin{align}
& U_\sR \rho_{\mathrm{ABE}}^{(\text{after Alice})} U_\sR^\dagger
\nonumber \\
&= \frac{1}{2} \bkb{0}_\sA \otimes P \left(
\ket{0, 1}_\sB \sum_{m_0 \ge 0} \ket{g_{0,m_0}^{0,1}}_\sE \right.
\nonumber \\
&+ \left.
\ket{1, 0}_\sB \sum_{m_0 \ge 0} \ket{g_{0,m_0}^{1,0}}_\sE +
\ket{0, 0}_\sB \sum_{m_0 \ge 0} \ket{g_{0,m_0}^{0,0}}_\sE \right)
\nonumber
\end{align}
\begin{align}
&+ \frac{1}{2} \bkb{1}_\sA \otimes P \left(
\ket{0, 1}_\sB \sum_{m_1 \ge 0} \ket{g_{m_1,0}^{0,1}}_\sE \right.
\nonumber \\
&+ \left.
\ket{1, 0}_\sB \sum_{m_1 \ge 0} \ket{g_{m_1,0}^{1,0}}_\sE +
\ket{0, 0}_\sB \sum_{m_1 \ge 0} \ket{g_{m_1,0}^{0,0}}_\sE \right).
\label{eq:state-before-b-q}
\end{align}
To simplify notation, we define the following states
in subsystem $\sE$:
\begin{eqnarray}
\ket{E_0}_\sE &\triangleq& \frac{1}{\sqrt{2}}
\sum_{m_0 \ge 0} \ket{g_{0,m_0}^{0,1}}_\sE,\nonumber\\
\ket{E_1}_\sE &\triangleq& \frac{1}{\sqrt{2}}
\sum_{m_0 \ge 0} \ket{g_{0,m_0}^{1,0}}_\sE,\nonumber\\
\ket{E_2}_\sE &\triangleq& \frac{1}{\sqrt{2}}
\sum_{m_1 \ge 0} \ket{g_{m_1,0}^{0,1}}_\sE,\nonumber\\
\ket{E_3}_\sE &\triangleq& \frac{1}{\sqrt{2}}
\sum_{m_1 \ge 0} \ket{g_{m_1,0}^{1,0}}_\sE,\label{eq:E-states}
\end{eqnarray}
so Eq.~\eqref{eq:state-before-b-q} becomes:
\begin{align}
& U_\sR \rho_{\mathrm{ABE}}^{(\text{after Alice})} U_\sR^\dagger
\nonumber \\
&= \bkb{0}_\sA \otimes P \left(
\ket{0, 1}_\sB \ket{E_0}_\sE + \ket{1, 0}_\sB \ket{E_1}_\sE \right.
\nonumber \\
&+ \left.
\ket{0, 0}_\sB \frac{1}{\sqrt{2}}
\sum_{m_0 \ge 0} \ket{g_{0,m_0}^{0,0}}_\sE \right)
\nonumber \\
&+ \bkb{1}_\sA \otimes P \left(
\ket{0, 1}_\sB \ket{E_2}_\sE + \ket{1, 0}_\sB \ket{E_3}_\sE \right.
\nonumber \\
&+ \left.\ket{0, 0}_\sB \frac{1}{\sqrt{2}}
\sum_{m_1 \ge 0} \ket{g_{m_1,0}^{0,0}}_\sE \right).
\label{eq:state-before-b}
\end{align}

\subsubsection{Standard ``Raw Key'' Rounds:
Bob Chooses the Computational Basis}
Now, Bob measures his subsystem in the computational basis
$\{\bket{0}, \bket{1}\}$, and his raw key bit
is simply his measurement result (``$0$'' or ``$1$'').
Conditioning on Bob detecting a photon
(namely, measuring $\ket{0,1}_\sB$ or $\ket{1,0}_\sB$),
the final \textit{normalized} state of the joint system
after Bob's measurement is:
\begin{align}\label{eq:final-state}
\rho_{\mathrm{ABE}} = \frac{1}{M}(
& \bkb{00}_{\mathrm{AB}} \otimes \kb{E_0}_\sE \nonumber \\
+&\bkb{01}_{\mathrm{AB}} \otimes \kb{E_1}_\sE \nonumber \\
+&\bkb{10}_{\mathrm{AB}} \otimes \kb{E_2}_\sE \nonumber \\
+&\bkb{11}_{\mathrm{AB}} \otimes \kb{E_3}_\sE),
\end{align}
where $M$ is a normalization term, which is computed below.

Eq.~\eqref{eq:final-state} confirms that,
as written in Table~\ref{table_probs}:
\begin{align}
\bk{E_0}_\sE = \pr&\left(\text{Alice gets raw key bit $0$,}\right.
\nonumber \\
&\left.\text{and Bob gets raw key bit $0$}\right), \\
\bk{E_1}_\sE = \pr&\left(\text{Alice gets raw key bit $0$,}\right.
\nonumber \\
&\left.\text{and Bob gets raw key bit $1$}\right), \\
\bk{E_2}_\sE = \pr&\left(\text{Alice gets raw key bit $1$,}\right.
\nonumber \\
&\left.\text{and Bob gets raw key bit $0$}\right), \\
\bk{E_3}_\sE = \pr&\left(\text{Alice gets raw key bit $1$,}\right.
\nonumber \\
&\left.\text{and Bob gets raw key bit $1$}\right).
\end{align}
In addition, we can compute the normalization term $M$:
\begin{eqnarray}
M &=& \sum_{i=0}^3 \bk{E_i}_\sE \label{eq:M} \\
&=& \pr(\text{both Alice and Bob
get raw key bits}) \nonumber \\
&=& \pr\left(\text{Alice observes no photon,}\right. \nonumber \\
&&\left.\text{and Bob observes a photon}\right). \nonumber
\end{eqnarray}

Notice that all these probabilities are \textit{observable} quantities:
Alice and Bob estimate $\bk{E_0}_\sE$, $\bk{E_1}_\sE$, $\bk{E_2}_\sE$,
$\bk{E_3}_\sE$, and $M$ during the classical post-processing stage
by testing a random subset of raw key bits.

\subsubsection{Mismatched ``Raw Key'' Rounds:
Bob Chooses the Hadamard Basis}
In this case, Bob measures his subsystem in the Hadamard basis
$\{\bket{+}, \bket{-}\}$.
Let us rewrite the state he measures, provided in
Eq.~\eqref{eq:state-before-b}, by substituting
$\ket{0,1}_\sB = \frac{\ket{+}_\sB + \ket{-}_\sB}{\sqrt{2}}$ and
$\ket{1,0}_\sB = \frac{\ket{+}_\sB - \ket{-}_\sB}{\sqrt{2}}$.
We get:
\begin{align}
& U_\sR \rho_{\mathrm{ABE}}^{(\text{after Alice})} U_\sR^\dagger
\nonumber \\
&= \bkb{0}_\sA \otimes P \left(
\ket{0, 1}_\sB \ket{E_0}_\sE + \ket{1, 0}_\sB \ket{E_1}_\sE \right.
\nonumber \\
& \left. + \ket{0, 0}_\sB \frac{1}{\sqrt{2}}
\sum_{m_0 \ge 0} \ket{g_{0,m_0}^{0,0}}_\sE \right)
\nonumber \\
&+ \bkb{1}_\sA \otimes P \left(
\ket{0, 1}_\sB \ket{E_2}_\sE + \ket{1, 0}_\sB \ket{E_3}_\sE \right.
\nonumber \\
& \left. + \ket{0, 0}_\sB \frac{1}{\sqrt{2}}
\sum_{m_1 \ge 0} \ket{g_{m_1,0}^{0,0}}_\sE \right)
\nonumber \\
&= \bkb{0}_\sA \otimes P \left( \frac{\ket{+}_\sB}{\sqrt{2}}
(\ket{E_0}_\sE + \ket{E_1}_\sE) + \cdots \right) \nonumber \\
&+ \bkb{1}_\sA \otimes P \left( \frac{\ket{+}_\sB}{\sqrt{2}}
(\ket{E_2}_\sE + \ket{E_3}_\sE) + \cdots \right),
\end{align}
where the remainders of the above terms (the ``$\cdots$'') are
irrelevant to our discussion.

We denote by $p_{0,+}$ the probability that
Alice gets the raw key bit $0$ and Bob observes $\ket{+}_\sB$
(see Table~\ref{table_probs}).
Similarly, we denote by $p_{1,+}$ the probability that
Alice gets the raw key bit $1$ and Bob observes $\ket{+}_\sB$.
These probabilities are:
\begin{eqnarray*}
p_{0,+} &=& \left|\frac{\ket{E_0}_\sE + \ket{E_1}_\sE}{\sqrt2}\right|^2
\\
&=& \frac{1}{2} \left(\bk{E_0}_\sE + \bk{E_1}_\sE
+ 2 \Re\braket{E_0|E_1}_\sE\right), \\
p_{1,+} &=& \left|\frac{\ket{E_2}_\sE + \ket{E_3}_\sE}{\sqrt2}\right|^2
\\
&=& \frac{1}{2} \left( \bk{E_2}_\sE + \bk{E_3}_\sE
+ 2 \Re\braket{E_2|E_3}_\sE\right).
\end{eqnarray*}
Therefore, we find:
\begin{align}
2 \Re\braket{E_0|E_1}_\sE &= 2 p_{0,+}
- \left(\bk{E_0}_\sE + \bk{E_1}_\sE\right),\label{eq:p0+}\\
2 \Re\braket{E_2|E_3}_\sE &= 2 p_{1,+}
- \left(\bk{E_2}_\sE + \bk{E_3}_\sE\right).\label{eq:p1+}
\end{align}

\subsection{``Test'' Rounds:
Alice Chooses the CTRL Operation}\label{subsec_testkey}
In ``test'' rounds, Eve sends to Alice her state
$\ket{\psi_0} \triangleq \sum_{\substack{m_1\ge 0\\m_0\ge 0}}
\ket{m_1, m_0}_\sB \ket{e_{m_1,m_0}}_\sE$
(see Eq.~\eqref{eq:init-state}), and Alice chooses
the CTRL operation---namely, Alice does nothing
(see Eq.~\eqref{eq:I-def}).
Then, Eve applies her second attack $U_\sR$
(see Eq.~\eqref{eq:reverse-attack-op-special}),
and the resulting quantum state is:
\begin{eqnarray}
U_\sR \ket{\psi_0} &=& \ket{0,1}_\sB
\sum_{\substack{m_1 \ge 0 \\ m_0 \ge 0}} \ket{g_{m_1,m_0}^{0,1}}_\sE
+ \ket{1,0}_\sB
\sum_{\substack{m_1 \ge 0 \\ m_0 \ge 0}} \ket{g_{m_1,m_0}^{1,0}}_\sE
\nonumber \\
&+& \ket{0,0}_\sB
\sum_{\substack{m_1 \ge 0 \\ m_0 \ge 0}} \ket{g_{m_1,m_0}^{0,0}}_\sE.
\label{eq:reflect-state}
\end{eqnarray}

\subsubsection{Standard ``Test'' Rounds:
Bob Chooses the Hadamard Basis}
Changing basis, whereby
$\ket{0,1}_\sB = \frac{\ket{+}_\sB + \ket{-}_\sB}{\sqrt{2}}$ and
$\ket{1,0}_\sB = \frac{\ket{+}_\sB - \ket{-}_\sB}{\sqrt{2}}$, we find:
\begin{eqnarray}
U_\sR \ket{\psi_0} &=& \frac{\ket{+}_\sB}{\sqrt{2}}
\left( \sum_{\substack{m_1 \ge 0 \\ m_0 \ge 0}}
\ket{g_{m_1,m_0}^{0,1}}_\sE
+ \sum_{\substack{m_1 \ge 0 \\ m_0 \ge 0}} \ket{g_{m_1,m_0}^{1,0}}_\sE
\right) \nonumber \\
&+& \cdots,\label{eq:reflect-state-x}
\end{eqnarray}
where the extra $\cdots$ term is irrelevant to our discussion.

Let $p_{+,+}$ be the probability that Bob observes $\ket{+}_\sB$
(see Table~\ref{table_probs}).
From Eq.~\eqref{eq:reflect-state-x} we deduce:
\begin{eqnarray}
p_{+,+} &=& \left|
\frac{1}{\sqrt{2}} \sum_{\substack{m_1 \ge 0 \\ m_0 \ge 0}}
\ket{g_{m_1,m_0}^{0,1}}_\sE +
\frac{1}{\sqrt{2}} \sum_{\substack{m_1 \ge 0 \\ m_0 \ge 0}}
\ket{g_{m_1,m_0}^{1,0}}_\sE \right|^2\nonumber \\
&=& \left|\left(\ket{E_0}_\sE + \ket{E_2}_\sE
- \ket{g_0}_\sE + \ket{h_0}_\sE\right)\right. \nonumber \\
&+& \left.\left(\ket{E_1}_\sE + \ket{E_3}_\sE
- \ket{g_1}_\sE + \ket{h_1}_\sE\right)\right|^2\nonumber \\
&=& \left|\ket{E_0}_\sE + \ket{E_2}_\sE
- \ket{g_0}_\sE + \ket{h_0}_\sE\right|^2\nonumber \\
&+& \left|\ket{E_1}_\sE + \ket{E_3}_\sE
- \ket{g_1}_\sE + \ket{h_1}_\sE\right|^2\nonumber \\
&+& 2 \Re\left[\left(\bra{E_0}_\sE + \bra{E_2}_\sE
- \bra{g_0}_\sE + \bra{h_0}_\sE\right)\right.\nonumber\\
&\times& \left.\left(\ket{E_1}_\sE + \ket{E_3}_\sE
- \ket{g_1}_\sE + \ket{h_1}_\sE\right)\right],\label{eq:pr-pp}
\end{eqnarray}
where we define:
\begin{eqnarray}
\ket{g_0}_\sE &\triangleq& \frac{1}{\sqrt{2}} \ket{g_{0,0}^{0,1}}_\sE,
\nonumber \\
\ket{g_1}_\sE &\triangleq& \frac{1}{\sqrt{2}} \ket{g_{0,0}^{1,0}}_\sE,
\nonumber \\
\ket{h_0}_\sE &\triangleq& \frac{1}{\sqrt{2}}
\sum_{\substack{m_1 \ge 1 \\ m_0 \ge 1}} \ket{g_{m_1,m_0}^{0,1}}_\sE,
\nonumber \\
\ket{h_1}_\sE &\triangleq& \frac{1}{\sqrt{2}}
\sum_{\substack{m_1 \ge 1 \\ m_0 \ge 1}} \ket{g_{m_1,m_0}^{1,0}}_\sE,
\label{eq:g-h-def}
\end{eqnarray}
and we remember from Eq.~\eqref{eq:E-states} that:
\begin{eqnarray*}
\ket{E_0}_\sE &\triangleq& \frac{1}{\sqrt{2}}
\sum_{m_0 \ge 0} \ket{g_{0,m_0}^{0,1}}_\sE, \\
\ket{E_1}_\sE &\triangleq& \frac{1}{\sqrt{2}}
\sum_{m_0 \ge 0} \ket{g_{0,m_0}^{1,0}}_\sE, \\
\ket{E_2}_\sE &\triangleq& \frac{1}{\sqrt{2}}
\sum_{m_1 \ge 0} \ket{g_{m_1,0}^{0,1}}_\sE, \\
\ket{E_3}_\sE &\triangleq& \frac{1}{\sqrt{2}}
\sum_{m_1 \ge 0} \ket{g_{m_1,0}^{1,0}}_\sE.
\end{eqnarray*}

\subsubsection{Mismatched ``Test'' Rounds:
Bob Chooses the Computational Basis}
In this case, we denote by $\pctrl{0}$ the probability of
Bob observing $\ket{0,1}_\sB$ (see Table~\ref{table_probs}).
From Eq.~\eqref{eq:reflect-state},
we find (similarly to the computation of $p_{+,+}$):
\begin{eqnarray}
\pctrl{0} &=& \left|\sum_{\substack{m_1 \ge 0 \\ m_0 \ge 0}}
\ket{g_{m_1,m_0}^{0,1}}_\sE\right|^2\nonumber\\
&=& 2 \left|\ket{E_0}_\sE + \ket{E_2}_\sE
- \ket{g_0}_\sE + \ket{h_0}_\sE\right|^2.\label{eq:pctrl0}
\end{eqnarray}
Similarly, denoting by $\pctrl{1}$ the probability of Bob observing
$\ket{1,0}_\sB$, we find:
\begin{eqnarray}
\pctrl{1} &=& \left|\sum_{\substack{m_1 \ge 0 \\ m_0 \ge 0}}
\ket{g_{m_1,m_0}^{1,0}}_\sE\right|^2\nonumber\\
&=& 2 \left|\ket{E_1}_\sE + \ket{E_3}_\sE
- \ket{g_1}_\sE + \ket{h_1}_\sE\right|^2.\label{eq:pctrl1}
\end{eqnarray}

\subsection{``SWAP-ALL'' Rounds:
Alice Chooses the SWAP-ALL Operation,
and Bob Chooses the Computational Basis}\label{subsec_swap_all}
\subsubsection{The Probability of a ``Double-Click'' Event:
Used for Upper-Bounding $\bk{h_0}_\sE$ and $\bk{h_1}_\sE$}
In ``SWAP-ALL'' rounds, Eve sends to Alice the initial state
$\ket{\psi_0} \triangleq \sum_{\substack{m_1\ge 0\\m_0\ge 0}}
\ket{m_1, m_0}_\sB \ket{e_{m_1,m_0}}_\sE$ described in
Eq.~\eqref{eq:init-state}, and Alice chooses the SWAP-ALL operation
defined in Eq.~\eqref{eq:S-def}, which essentially means that Alice
measures subsystem $\sB$ and sends a vacuum state towards Bob.

Let us denote by $\pdouble$ the probability that Alice observes a
``double-click'' event (detecting a photon in \textit{both} modes
$\bket{0}$ and $\bket{1}$)---namely,
that she measures a state $\ket{m_1, m_0}_{\sA_{\mathrm{anc}}}$
where $m_1, m_0 \ge 1$
(see Table~\ref{table_probs}). This probability is easily found to be:
\[
\pdouble =
\sum_{\substack{m_1 \ge 1 \\ m_0 \ge 1}}\bk{e_{m_1,m_0}}_\sE.
\]

We can thus prove the following Lemma:
\begin{lemma}\label{lemma}
$\bk{h_0}_\sE \le \frac{1}{2} \pdouble$ and
$\bk{h_1}_\sE \le \frac{1}{2} \pdouble$,
where $\ket{h_0}_\sE,\ket{h_1}_\sE$
were defined in Eq.~\eqref{eq:g-h-def}.
\end{lemma}
\begin{proof}
Let us define the non-normalized state $\ket{\zeta}$ as:
\[
\ket{\zeta} \triangleq \frac{1}{\sqrt{2}}
\sum_{\substack{m_1 \ge 1 \\ m_0 \ge 1}}
\ket{m_1, m_0}_\sB \ket{e_{m_1,m_0}}_\sE.
\]
(We use the state $\ket{\zeta}$ only for this algebraic proof;
it does not appear in the protocol.)

Clearly:
\[
\bk{\zeta} = \frac{1}{2}
\sum_{\substack{m_1 \ge 1 \\ m_0 \ge 1}}\bk{e_{m_1,m_0}}_\sE
= \frac{1}{2} \pdouble.
\]

Applying $U_\sR$ (see Eq.~\eqref{eq:reverse-attack-op-special}),
the state $\ket{\zeta}$ evolves to:
\begin{eqnarray*}
U_\sR\ket{\zeta} &=& \frac{1}{\sqrt{2}}
\sum_{\substack{m_1 \ge 1 \\ m_0 \ge 1}}
\left( \ket{0,1}_\sB \ket{g_{m_1,m_0}^{0,1}}_\sE \right. \\
&+& \left. \ket{1,0}_\sB \ket{g_{m_1,m_0}^{1,0}}_\sE
+ \ket{0,0}_\sB \ket{g_{m_1,m_0}^{0,0}}_\sE \right) \\
&=& \ket{0,1}_\sB \ket{h_0}_\sE + \ket{1,0}_\sB \ket{h_1}_\sE
+ \ket{0,0}_\sB \ket{h_{\mathrm{vac}}}_\sE
\end{eqnarray*}
(where $\ket{h_0}_\sE,\ket{h_1}_\sE$ were defined in
Eq.~\eqref{eq:g-h-def},
and $\ket{h_{\mathrm{vac}}}_\sE \triangleq \frac{1}{\sqrt{2}}
\sum_{\substack{m_1 \ge 1 \\ m_0 \ge 1}} \ket{g_{m_1,m_0}^{0,0}}_\sE$).

By unitarity of $U_\sR$, we have:
\[
\frac{1}{2} \pdouble = \bk{\zeta} = \bk{h_0}_\sE + \bk{h_1}_\sE
+ \bk{h_{\mathrm{vac}}}_\sE,
\]
which implies that $\bk{h_0}_\sE + \bk{h_1}_\sE
\le \frac{1}{2} \pdouble$.
Since both $\bk{h_0}_\sE$ and $\bk{h_1}_\sE$ are non-negative,
this implies $\bk{h_0}_\sE \le \frac{1}{2} \pdouble$
and $\bk{h_1}_\sE \le \frac{1}{2} \pdouble$, as we wanted.
\end{proof}

\subsubsection{The Probability of a ``Creation'' Event:
Used for Computing $\bk{g_0}_\sE$ and $\bk{g_1}_\sE$}
Let $\pcreate{0}$ denote the probability that
Alice observes $\ket{0,0}_{\sA_{\mathrm{anc}}}$
(namely, a vacuum state) and
Bob observes $\ket{0,1}_\sB$ (see Table~\ref{table_probs}).
In this event, Eve ``creates'' (on the way from Alice to Bob)
a photon in the $\bket{0}$ mode that should not have existed.
(See~\cite{mirror_attack18} for examples of such attacks.)
Similarly, let $\pcreate{1}$ denote the probability that
Alice observes $\ket{0,0}_{\sA_{\mathrm{anc}}}$
and Bob observes $\ket{1,0}_\sB$.

After Eve sends the initial state
\[\ket{\psi_0} \triangleq \sum_{\substack{m_1\ge 0\\m_0\ge 0}}
\ket{m_1, m_0}_\sB \ket{e_{m_1,m_0}}_\sE\] described in
Eq.~\eqref{eq:init-state}, and after Alice applies the SWAP-ALL
operation defined in Eq.~\eqref{eq:S-def}, the resulting state is:
\[
\sum_{\substack{m_1 \ge 0 \\ m_0 \ge 0}}
\ket{m_1, m_0}_{\sA_{\mathrm{anc}}}
\ket{0, 0}_\sB \ket{e_{m_1,m_0}}_\sE.
\]
For computing the probabilities $\pcreate{0}$ and $\pcreate{1}$,
we need to analyze the term
where Alice observes $\ket{0,0}_{\sA_{\mathrm{anc}}}$---namely,
the term
$\ket{0, 0}_{\sA_{\mathrm{anc}}} \ket{0, 0}_\sB \ket{e_{0,0}}_\sE$.
Now, Eve's second attack applies the unitary operator $U_\sR$
(described in Eq.~\eqref{eq:reverse-attack-op-special})
to this non-normalized term, which gives the following final result:
\begin{align*}
&\ket{0, 0}_{\sA_{\mathrm{anc}}}
\otimes U_\sR \ket{0, 0}_\sB \ket{e_{0,0}}_\sE \\
&= \ket{0, 0}_{\sA_{\mathrm{anc}}}
\otimes \left[ \ket{0,1}_\sB \ket{g_{0,0}^{0,1}}_\sE
\right. \\
&+ \left. \ket{1,0}_\sB \ket{g_{0,0}^{1,0}}_\sE
+ \ket{0,0}_\sB \ket{g_{0,0}^{0,0}}_\sE \right].
\end{align*}
Since $\pcreate{0}$ is the probability that
Alice observes $\ket{0,0}_{\sA_{\mathrm{anc}}}$
and Bob observes $\ket{0,1}_\sB$
(and similarly for $\pcreate{1}$), we get, according to the definitions
of $\ket{g_0}_\sE,\ket{g_1}_\sE$ in Eq.~\eqref{eq:g-h-def}:
\begin{eqnarray}
\pcreate{0} &=& \bk{g_{0,0}^{0,1}}_\sE
= 2\bk{g_0}_\sE,\label{eq:pcreate0} \\
\pcreate{1} &=& \bk{g_{0,0}^{1,0}}_\sE
= 2\bk{g_1}_\sE.\label{eq:pcreate1}
\end{eqnarray}

\subsection{Deriving the Final Key Rate}\label{subsec_keyrate}
We remember that the final normalized state of the joint system
after Bob's measurement, in standard ``raw key'' rounds where
raw key bits \textit{are} generated, is,
according to Eq.~\eqref{eq:final-state}:
\begin{align*}
\rho_{\mathrm{ABE}} = \frac{1}{M}(
& \bkb{00}_{\mathrm{AB}} \otimes \kb{E_0}_\sE \nonumber \\
+&\bkb{01}_{\mathrm{AB}} \otimes \kb{E_1}_\sE \nonumber \\
+&\bkb{10}_{\mathrm{AB}} \otimes \kb{E_2}_\sE \nonumber \\
+&\bkb{11}_{\mathrm{AB}} \otimes \kb{E_3}_\sE).
\end{align*}

Theorem~1 from~\cite{QKD-Tom-Krawec-Arbitrary} allows us
to mathematically compute a bound
on the conditional von Neumann entropy $S(A|E)$
of $\rho_{\mathrm{ABE}}$, as follows:
\begin{eqnarray}
S(A|E) &\ge& \frac{\braket{E_0|E_0}_\sE + \braket{E_3|E_3}_\sE}{M}
\label{eq:SAE-bound} \\
&\times&
\left[ H_2\left(\frac{\bk{E_0}_\sE}{\bk{E_0}_\sE + \bk{E_3}_\sE}\right)
- H_2(\lambda_1) \right]\nonumber\\
&+& \frac{\braket{E_1|E_1}_\sE + \braket{E_2|E_2}_\sE}{M}
\nonumber \\
&\times&
\left[ H_2\left(\frac{\bk{E_1}_\sE}{\bk{E_1}_\sE + \bk{E_2}_\sE}\right)
- H_2(\lambda_2) \right],\nonumber
\end{eqnarray}
where:
\begin{eqnarray*}
\lambda_1 &\triangleq& \frac{1}{2} \\
&+& \frac{\sqrt{\left(\bk{E_0}_\sE - \bk{E_3}_\sE\right)^2
+ 4\Re^2\braket{E_0|E_3}_\sE}}
{2\left(\bk{E_0}_\sE+\bk{E_3}_\sE\right)},
\end{eqnarray*}
\begin{eqnarray*}
\lambda_2 &\triangleq& \frac{1}{2} \\
&+& \frac{\sqrt{\left(\bk{E_1}_\sE - \bk{E_2}_\sE\right)^2
+ 4\Re^2\braket{E_1|E_2}_\sE}}
{2\left(\bk{E_1}_\sE+\bk{E_2}_\sE\right)},\\
H_2(x) &\triangleq& -x \log_2(x) - (1-x) \log_2(1-x).
\end{eqnarray*}

Thus, to complete our proof of security, we only need bounds on
the quantities $\Re\braket{E_0|E_3}_\sE$ and $\Re\braket{E_1|E_2}_\sE$;
all the other parameters in the above expressions
($\bk{E_0}_\sE$, $\bk{E_1}_\sE$, $\bk{E_2}_\sE$, $\bk{E_3}_\sE$,
and $M$) are observable probabilities that appear in
Table~\ref{table_probs} and can be directly computed by Alice and Bob.

\begin{lemma}\label{lemma:constraint}
  The following constraint on Eve's quantum states holds:
  \begin{align}
&\Re\left(\braket{E_0|E_3}_\sE + \braket{E_1|E_2}_\sE\right)\nonumber\\
&\ge \frac{1}{2} p_{+,+} - p_{0,+} - p_{1,+}
- \frac{1}{4}(\pctrl{0}+\pctrl{1}) + \frac{1}{2} M \nonumber \\
&- \frac{1}{\sqrt{2}}\left(\sqrt{\pcreate{1}} + \sqrt{\pdouble}\right)
\left(\sqrt{\bk{E_0}_\sE} + \sqrt{\bk{E_2}_\sE}\right) \nonumber \\
&- \frac{1}{\sqrt{2}}\left(\sqrt{\pcreate{0}} + \sqrt{\pdouble}\right)
\left(\sqrt{\bk{E_1}_\sE} + \sqrt{\bk{E_3}_\sE}\right) \nonumber \\
&- \frac{1}{2}\left(\sqrt{\pcreate{0}} + \sqrt{\pdouble}\right)
\left(\sqrt{\pcreate{1}} + \sqrt{\pdouble}\right).
\end{align}
\end{lemma}
\begin{proof}
We expand Eq.~\eqref{eq:pr-pp} and substitute
Eqs.~\eqref{eq:p0+}--\eqref{eq:p1+}
and~\eqref{eq:pctrl0}--\eqref{eq:pctrl1}
(all appearing in Table~\ref{table_probs}) to find:
\begin{eqnarray}
p_{+,+} &=& \left|\ket{E_0}_\sE + \ket{E_2}_\sE
- \ket{g_0}_\sE + \ket{h_0}_\sE\right|^2\nonumber \\
&+& \left|\ket{E_1}_\sE + \ket{E_3}_\sE
- \ket{g_1}_\sE + \ket{h_1}_\sE\right|^2\nonumber \\
&+& 2 \Re\left[\left(\bra{E_0}_\sE + \bra{E_2}_\sE
- \bra{g_0}_\sE + \bra{h_0}_\sE\right)\right.\nonumber\\
&\times& \left.\left(\ket{E_1}_\sE + \ket{E_3}_\sE
- \ket{g_1}_\sE + \ket{h_1}_\sE\right)\right] \nonumber \\\nonumber\\
&=& \frac{1}{2} (\pctrl{0} + \pctrl{1}) \nonumber \\
&+& 2 \Re\left(\bra{E_0}_\sE + \bra{E_2}_\sE\right)
\left(\ket{E_1}_\sE + \ket{E_3}_\sE\right) \nonumber \\
&-& 2 \Re\left(\bra{E_0}_\sE + \bra{E_2}_\sE\right)
\left(\ket{g_1}_\sE - \ket{h_1}_\sE\right) \nonumber \\
&-& 2 \Re\left(\bra{g_0}_\sE - \bra{h_0}_\sE\right)
\left(\ket{E_1}_\sE + \ket{E_3}_\sE\right) \nonumber \\
&+& 2 \Re\left(\bra{g_0}_\sE - \bra{h_0}_\sE\right)
\left(\ket{g_1}_\sE - \ket{h_1}_\sE\right) \nonumber \\\nonumber\\
&=& \frac{1}{2} (\pctrl{0} + \pctrl{1}) \nonumber \\
&+& 2 p_{0,+} - \left(\bk{E_0}_\sE + \bk{E_1}_\sE\right)
+ 2 \Re \braket{E_0|E_3}_\sE \nonumber \\
&+& 2 p_{1,+} - \left(\bk{E_2}_\sE + \bk{E_3}_\sE\right)
+ 2 \Re \braket{E_1|E_2}_\sE \nonumber \\
&-& 2 \Re\left(\bra{E_0}_\sE + \bra{E_2}_\sE\right)
\left(\ket{g_1}_\sE - \ket{h_1}_\sE\right) \nonumber \\
&-& 2 \Re\left(\bra{g_0}_\sE - \bra{h_0}_\sE\right)
\left(\ket{E_1}_\sE + \ket{E_3}_\sE\right) \nonumber \\
&+& 2 \Re\left(\bra{g_0}_\sE - \bra{h_0}_\sE\right)
\left(\ket{g_1}_\sE - \ket{h_1}_\sE\right).
\end{eqnarray}

From this, we easily find (substituting Eq.~\eqref{eq:M},
which appears in Table~\ref{table_probs}):
\begin{align}
&\Re\left(\braket{E_0|E_3}_\sE + \braket{E_1|E_2}_\sE\right)\nonumber\\
&= \frac{1}{2} p_{+,+} - p_{0,+} - p_{1,+}
- \frac{1}{4} (\pctrl{0} + \pctrl{1}) + \frac{1}{2} M \nonumber \\
&+ \Re\left(\bra{g_1}_\sE - \bra{h_1}_\sE\right)
\left(\ket{E_0}_\sE + \ket{E_2}_\sE\right) \nonumber \\
&+ \Re\left(\bra{g_0}_\sE - \bra{h_0}_\sE\right)
\left(\ket{E_1}_\sE + \ket{E_3}_\sE\right) \nonumber \\
&- \Re\left(\bra{g_0}_\sE - \bra{h_0}_\sE\right)
\left(\ket{g_1}_\sE - \ket{h_1}_\sE\right).
\end{align}
The Cauchy-Schwarz inequality, Lemma~\ref{lemma},
and Eqs.~\eqref{eq:pcreate0}--\eqref{eq:pcreate1}
(all appearing in Table~\ref{table_probs})
complete the proof.
\end{proof}

Taken together, the above proof derives a lower bound on $S(A|E)$
for a raw-key generation round, and this bound is based only
on observable statistics from Table~\ref{table_probs}.
The Devetak-Winter key rate equation~\cite{DW05} (which says that the key rate
of a QKD protocol under collective attacks is the difference $S(A|E) - H(A|B)$)
then completes our proof of Theorem~\ref{thm:main-thm}.

To actually evaluate our bound on $S(A|E)$,
we will simply minimize Eq.~\eqref{eq:SAE-bound}
with respect to the condition outlined in Lemma~\ref{lemma:constraint}
and the following conditions (resulting from the Cauchy-Schwarz inequality):
\begin{eqnarray}
\left|\Re\braket{E_0|E_3}_\sE\right|
&\le& \sqrt{\bk{E_0}_\sE\cdot\bk{E_3}_\sE},\label{eq:E0E3-bound}\\
\left|\Re\braket{E_1|E_2}_\sE\right|
&\le& \sqrt{\bk{E_1}_\sE\cdot\bk{E_2}_\sE}.\label{eq:E1E2-bound}
\end{eqnarray}

In addition, we need to compute the expression $H(A|B)$ (needed in Theorem \ref{thm:main-thm}):
\begin{equation}
H(A|B) = H(AB) - H(B),
\end{equation}
where:
\begin{align}
& H(AB) \nonumber \\
&= H \left( \frac{\bk{E_0}_\sE}{M}, \frac{\bk{E_1}_\sE}{M},
\frac{\bk{E_2}_\sE}{M}, \frac{\bk{E_3}_\sE}{M} \right), \\
& H(B) \nonumber \\
&= H \left( \frac{\bk{E_0}_\sE + \bk{E_2}_\sE}{M},
\frac{\bk{E_1}_\sE + \bk{E_3}_\sE}{M} \right), \nonumber \\
& H(x_1, \ldots, x_k) \triangleq -\sum_{j=1}^k x_j \log_2(x_j).
\end{align}

\subsection{Algorithm for Computing the Key Rate}\label{subsec_algo}
The following algorithm allows us to compute the key rate for
any noise model and experimental data:

\begin{enumerate}
\item Estimate all probabilities and inner products
listed in Table~\ref{table_probs}.
(All these probabilities can be computed
by Alice and Bob in the classical post-processing stage.)
\item Compute the minimal value of the lower bound for $S(A|E)$
presented in Eq.~\eqref{eq:SAE-bound}, which is copied here:
\begin{eqnarray}
S(A|E) &\ge& \frac{\braket{E_0|E_0}_\sE + \braket{E_3|E_3}_\sE}{M}
\label{eq:SAE-bound_algo} \\
&\times&
\left[ H_2\left(\frac{\bk{E_0}_\sE}{\bk{E_0}_\sE + \bk{E_3}_\sE}\right)
- H_2(\lambda_1) \right]\nonumber\\
&+& \frac{\braket{E_1|E_1}_\sE + \braket{E_2|E_2}_\sE}{M}
\nonumber \\
&\times&
\left[ H_2\left(\frac{\bk{E_1}_\sE}{\bk{E_1}_\sE + \bk{E_2}_\sE}\right)
- H_2(\lambda_2) \right],\nonumber
\end{eqnarray}
where
\begin{eqnarray*}
\lambda_1 &\triangleq& \frac{1}{2} \\
&+& \frac{\sqrt{\left(\bk{E_0}_\sE - \bk{E_3}_\sE\right)^2
+ 4\Re^2\braket{E_0|E_3}_\sE}}
{2\left(\bk{E_0}_\sE+\bk{E_3}_\sE\right)},\\
\lambda_2 &\triangleq& \frac{1}{2} \\
&+& \frac{\sqrt{\left(\bk{E_1}_\sE - \bk{E_2}_\sE\right)^2
+ 4\Re^2\braket{E_1|E_2}_\sE}}
{2\left(\bk{E_1}_\sE+\bk{E_2}_\sE\right)},\\
H_2(x) &\triangleq& -x \log_2(x) - (1-x) \log_2(1-x),
\end{eqnarray*}
where the minimum is taken over $\Re \braket{E_0|E_3}_\sE$ and
$\Re \braket{E_1|E_2}_\sE$, subject to the three following constraints:
\begin{align}
&\Re\left(\braket{E_0|E_3}_\sE + \braket{E_1|E_2}_\sE\right)\nonumber\\
&\ge \frac{1}{2} p_{+,+} - p_{0,+} - p_{1,+}
- \frac{1}{4}(\pctrl{0}+\pctrl{1}) + \frac{1}{2} M \nonumber \\
&- \frac{1}{\sqrt{2}}\left(\sqrt{\pcreate{1}} + \sqrt{\pdouble}\right)\nonumber\\
&\times
\left(\sqrt{\bk{E_0}_\sE} + \sqrt{\bk{E_2}_\sE}\right) \nonumber \\
&- \frac{1}{\sqrt{2}}\left(\sqrt{\pcreate{0}} + \sqrt{\pdouble}\right)\nonumber\\
&\times\left(\sqrt{\bk{E_1}_\sE} + \sqrt{\bk{E_3}_\sE}\right) \nonumber \\
&- \frac{1}{2}\left(\sqrt{\pcreate{0}} + \sqrt{\pdouble}\right)\nonumber\\
&\times\left(\sqrt{\pcreate{1}} + \sqrt{\pdouble}\right),\label{eq:lbound} \\
&|\Re\braket{E_0|E_3}_\sE| \le \sqrt{\bk{E_0}_\sE\cdot\bk{E_3}_\sE}
\label{eq:csbound1},\\
&|\Re\braket{E_1|E_2}_\sE| \le \sqrt{\bk{E_1}_\sE\cdot\bk{E_2}_\sE}
\label{eq:csbound2}.
\end{align}
Note that we evaluate the minimum
because we assume the worst-case scenario---namely, that Eve chooses
her attack so as to minimize $S(A|E)$
(and, thus, minimize the key rate $r$).

In practice, we can minimize over a single parameter
(say, $\Re \braket{E_1|E_2}_\sE$), and take the other one
($\Re\braket{E_0|E_3}_\sE$) as the right-hand-side of
Eq.~\eqref{eq:lbound}, minus the free parameter
$\Re \braket{E_1|E_2}_\sE$ (but not less than $0$).
This will give us the minimum, because for any given value of
$\Re \braket{E_1|E_2}_\sE$, it is beneficial for Eve to have the
smallest possible (non-negative) value of $\Re\braket{E_0|E_3}_\sE$.

For our evaluations, we performed this minimization by simply discretizing
the search space and evaluating our bound on the entropy
at all points in the space for computing the minimum.
We also confirmed these results using Mathematica's \texttt{NMinimize} function.
\item Compute $H(A|B)$ using the observed parameters:
\begin{align}
& H(A|B) = H(AB) - H(B) \nonumber \\
&= H \left( \frac{\bk{E_0}_\sE}{M}, \frac{\bk{E_1}_\sE}{M},
\frac{\bk{E_2}_\sE}{M}, \frac{\bk{E_3}_\sE}{M} \right) \nonumber \\
&- H \left( \frac{\bk{E_0}_\sE + \bk{E_2}_\sE}{M},
\frac{\bk{E_1}_\sE + \bk{E_3}_\sE}{M} \right),\label{eq:HAB_algo}
\end{align}
where:
\begin{equation}
H(x_1, \ldots, x_k) \triangleq -\sum_{j=1}^k x_j \log_2(x_j).
\end{equation}
\item Find the final key rate expression, using the Devetak-Winter
key rate formula~\cite{DW05}:
\begin{equation}
r = S(A|E) - H(A|B).
\end{equation}
\end{enumerate}

This process is summarized in Algorithm~\ref{alg:1}.

\begin{algorithm}
  \KwIn{All observable probabilities listed in Table~\ref{table_probs}.}
  \KwOut{Lower bound on the key rate of the protocol.}

  Initialize the variable $\texttt{lowestAE} \leftarrow \infty$.\\
  Compute all probabilities listed in Table~\ref{table_probs}
  using the protocol's statistics observed by Alice and Bob.\\
  \tcc{Next, minimize $S(A|E)$ by minimizing
  the lower bound in Eq.~\eqref{eq:SAE-bound_algo}.}
  \For{all possible $\Re\braket{E_1|E_2}_\sE$ subject to Eq.~\eqref{eq:E1E2-bound}}
    {
      Compute a lower bound on $\Re\braket{E_0|E_3}_\sE$
      using Eq.~\eqref{eq:lbound} and subject to Eq.~\eqref{eq:E0E3-bound}.

      Compute a lower bound on $S(A|E)$ using Eq.~\eqref{eq:SAE-bound_algo}.

      If this determined bound is lower than the existing value
      of $\texttt{lowestAE}$, save it in $\texttt{lowestAE}$.
    }

    Compute $H(A|B)$ using Eq.~\eqref{eq:HAB_algo},
    and put the result in variable $\texttt{AB}$.

    \Return{the difference value $\texttt{lowestAE} - \texttt{AB}$}
    \caption{Compute a Lower Bound for
    $\texttt{rate} = S(A|E) - H(A|B)$.}\label{alg:1}
\end{algorithm}

\section{Examples}
The key rate bounds we found in Section~\ref{sec_proof}
work in a wide range of scenarios, and they can be evaluated
for all the possible values
of all probabilities in Table~\ref{table_probs}.
We would now like to evaluate our bounds for two concrete scenarios,
that are easily comparable with attacks
on other QKD and SQKD protocols.

\subsection{First Scenario: Single-Photon Attacks without Losses}
In the first scenario,
let us assume that Bob has a perfect qubit source
(no multi-photon pulses) and there are no photon losses.
Furthermore, let us assume that Eve does not perform
a multi-qubit attack at all (not even in her \textit{first} attack).
In this scenario, the only free parameters are the noises $Q_\sZ,Q_\sX$
in the channel: $Q_\sZ$ is the probability that a $\ket{0,1}_\sB$ state
is flipped into $\ket{1,0}_\sB$ (and vice versa) in ``raw key'' rounds,
and $Q_\sX$ is the probability that a $\ket{+}_\sB$ state is flipped
into $\ket{-}_\sB$ in ``test'' rounds.

We consider the following noise model:
\begin{itemize}
\item In the ``raw key'' rounds, we consider that \textit{both}
the forward channel (from Bob to Alice) and
the reverse channel (from Alice to Bob)
are depolarizing channels with error $Q_\sZ$, as follows:
\begin{equation}
\mathcal{E}_{Q_\sZ}(\rho) = (1 - 2 Q_\sZ) \rho
+ 2 Q_\sZ \cdot \frac{I_2}{2}.
\end{equation}
\item In the ``test'' rounds, we consider that the whole channel
(from Bob to Alice and back to Bob;
notice that Alice does nothing in such rounds)
is a depolarizing channel with error $Q_\sX$, as follows:
\begin{equation}
\mathcal{E}_{Q_\sX}(\rho) = (1 - 2 Q_\sX) \rho
+ 2 Q_\sX \cdot \frac{I_2}{2}.
\end{equation}
\end{itemize}

Here, in the forward attack, Eve always replaces Bob's original state
$\ket{0,1}_{\sx,\sB} \triangleq
\frac{\ket{0,1}_\sB + \ket{1,0}_\sB}{\sqrt{2}}$ by the following state
(a special case of Eq.~\eqref{eq:init-state}):
\begin{equation}
\ket{\psi_0} = \ket{0,1}_\sB \ket{e_{0,1}}_\sE
+ \ket{1,0}_\sB \ket{e_{1,0}}_\sE,
\end{equation}
with $\bk{e_{0,1}}_\sE = \bk{e_{1,0}}_\sE = \frac{1}{2}$.

\subsection{Second Scenario: Single-Photon Attacks \textit{with} Losses}
In the second scenario, our noise model remains identical
to the first scenario, except two modifications:
\begin{itemize}
\item In the forward channel (from Bob to Alice),
a loss occurs with probability $p_\ell^\sF$;
if it \textit{does not} occur, the original noise model is applied.
\item In the reverse channel (from Alice to Bob),
a loss occurs with probability $p_\ell^\sR$;
if it \textit{does not} occur, the original noise model is applied.
\end{itemize}
We assume, in particular, that a loss is \textit{final}:
if a loss occurs in the forward channel, no photon will ever
be observed in this round by either Alice or Bob.

\subsection{Evaluation Results}
In Table~\ref{table_examples} we evaluate all probabilities
in both scenarios.

\begin{table*}[htpb]
\caption{Computing all probabilities in
Table~\ref{table_probs} for both examples (both scenarios).}
\label{table_examples}
\centering
\begin{tabular}{rcc}
\hline\hline
\textbf{Probability} & \textbf{Single-Photon without Losses}
& \textbf{Single-Photon with Losses} \\ \hline
$\bk{E_0}_\sE = \bk{E_3}_\sE =$ & $\frac{1}{4} (1 - Q_\sZ)$
& $\frac{1}{4} (1 - p_\ell^\sF) (1 - p_\ell^\sR) (1 - Q_\sZ)$ \\
$\bk{E_1}_\sE = \bk{E_2}_\sE =$ & $\frac{1}{4} Q_\sZ$
& $\frac{1}{4} (1 - p_\ell^\sF) (1 - p_\ell^\sR) Q_\sZ$ \\ \hline
$M =$ & $\frac{1}{2}$ & $\frac{1}{2} (1 - p_\ell^\sF) (1 - p_\ell^\sR)$
\\ \hline
$p_{0,+} = p_{1,+} =$ & $\frac{1}{8}$
& $\frac{1}{8} (1 - p_\ell^\sF) (1 - p_\ell^\sR)$ \\ \hline
$p_{+,+} =$ & $1 - Q_\sX$
& $(1 - p_\ell^\sF) (1 - p_\ell^\sR) (1 - Q_\sX)$ \\ \hline
$\pctrl{0} = \pctrl{1} =$ & $\frac{1}{2}$
& $\frac{1}{2} (1 - p_\ell^\sF) (1 - p_\ell^\sR)$ \\ \hline
$\pdouble =$ & $0$ & $0$ \\ \hline
$\pcreate{0} = \pcreate{1} =$ & $0$ & $0$ \\
\hline\hline
\end{tabular}
\end{table*}

\paragraph{First scenario---single-photon attacks without losses}
Substituting the probabilities from Table~\ref{table_examples}
in Eqs.~\eqref{eq:lbound}--\eqref{eq:csbound2},
we find the three constraints to be:
\begin{eqnarray}
\Re\left(\braket{E_0|E_3}_\sE + \braket{E_1|E_2}_\sE\right) &\ge&
\frac{1}{4} - \frac{1}{2} Q_\sX, \\
\left|\Re\braket{E_0|E_3}_\sE\right| &\le& \frac{1}{4} (1 - Q_\sZ), \\
\left|\Re\braket{E_1|E_2}_\sE\right| &\le& \frac{1}{4} Q_\sZ.
\end{eqnarray}
As explained in Subsection~\ref{subsec_algo},
we numerically find the minimal value of
the key rate expression $r = S(A|E) - H(A|B)$
for various values of $Q_{\sZ,\sX}$
by using the lower bound on $S(A|E)$
presented in Eq.~\eqref{eq:SAE-bound_algo},
which is evaluated under the three above constraints
on the values of $\Re\braket{E_0|E_3}_\sE$
and $\Re\braket{E_1|E_2}_\sE$.
This numerical optimization yields the graph
shown in Fig.~\ref{fig:keyrate-perfect}, presenting two cases:
\begin{itemize}
\item In the \textit{dependent} noise model, where the error rates
$Q_\sX$ and $Q_\sZ$ are identical (namely, $Q_\sX = Q_\sZ$),
we recover the asymptotic BB84 noise tolerance of $11\%$.
\item In the \textit{independent} noise model,
where the two-way channel is modeled as two independent
depolarizing channels (namely, $Q_\sX = 2Q_\sZ(1 - Q_\sZ)$),
the maximal (asymptotic) noise tolerance is $7.9\%$.
\end{itemize}
Interestingly, both values agree with the values
found in~\cite{QKD-Tom-Krawec-Arbitrary} for
the original ``QKD with Classical Bob'' SQKD protocol~\cite{cbob07}.

In both scenarios, because the Mirror protocol is two-way,
we compare it to \textit{two} copies of BB84 performed from Alice to Bob;
this is a common comparison for two-way protocols
(see, for example,~\cite{renner-twoway13}).
The key rate of two copies of BB84 is $2 (1 - 2 H_2(p))$---namely,
twice the original key rate of BB84.

\begin{figure}
\includegraphics[width=0.5\textwidth]{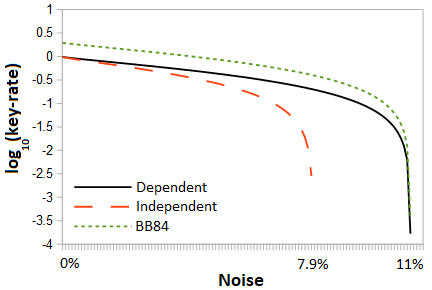}
\caption{A graph of the final key rate
versus the noise level of the Mirror protocol in
the first scenario (single-photon attacks without losses),
for dependent ($Q_\sX = Q_\sZ$)
and independent ($Q_\sX = 2Q_\sZ(1 - Q_\sZ)$) noise models,
compared to two copies of BB84.}\label{fig:keyrate-perfect}
\end{figure}

\paragraph{Second scenario---single-photon attacks with losses}
Substituting the probabilities from Table~\ref{table_examples}
in Eqs.~\eqref{eq:lbound}--\eqref{eq:csbound2},
we find the three constraints to be:
\begin{align}
&\Re\left(\braket{E_0|E_3}_\sE + \braket{E_1|E_2}_\sE\right)\nonumber\\
&\ge (1 - p_\ell^\sF) (1 - p_\ell^\sR)
\left(\frac{1}{4} - \frac{1}{2} Q_\sX\right), \\
&\left|\Re\braket{E_0|E_3}_\sE\right|\nonumber\\
&\le \frac{1}{4} (1 - p_\ell^\sF) (1 - p_\ell^\sR) (1 - Q_\sZ),
\end{align}
\begin{align}
&\left|\Re\braket{E_1|E_2}_\sE\right|\nonumber\\
&\le \frac{1}{4} (1 - p_\ell^\sF) (1 - p_\ell^\sR) Q_\sZ.
\end{align}
The numerical analysis for this scenario is similar
to the previous one.
However, here we must also model the loss rates,
so we consider a fiber channel with loss rates
$p_\ell^{\sF,\sR} = 1 - 10^{-\alpha \ell}$
(where $\alpha$ is the loss coefficient, and $\ell$ is measured in kilometers).
We consider two examples of fiber lengths:
$\ell = 10 \mathrm{km}$ and $\ell = 50 \mathrm{km}$.
Results are presented in Fig.~\ref{fig:keyrate-lossy}.

These evaluations lead to several observations---most notably, the observation
that the Mirror protocol is more sensitive to loss than BB84
even in the single photon case: increasing the fiber length
from $\ell = 10 \mathrm{km}$ to $\ell = 50 \mathrm{km}$
causes a significant drop in key rate.
We also note that the key rate of the Mirror protocol at only $10 \mathrm{km}$
coincides with that of BB84 at $50 \mathrm{km}$. This seems to indicate,
not surprisingly, that BB84 outperforms the Mirror protocol under loss.
There are many reasons for this. First, note that each photon in Mirror
travels twice the distance compared to BB84: while we are comparing Mirror
with two copies of BB84, these copies are treated independently and,
thus, for a single bit to be produced from these two copies, it is sufficient
for one of the photons to survive transmission without being lost
(over a fiber of length $\ell$). On the other hand, in the Mirror protocol,
the photon must travel through \textit{both} channels without loss
(a total fiber length of $2\ell$) for a single bit to be produced from a round.
Second, in Mirror Eve has two opportunities to attack, which gives her
a bigger attack strategy space for any given loss level.
Finally, our security analysis against loss may not be as tight as
our analysis against noise (where, as seen in Fig.~\ref{fig:keyrate-perfect},
the Mirror protocol performs similarly to BB84 under a lossless but noisy
channel). This is, to our knowledge, the first security proof for the Mirror
protocol against loss, and future improvements may exist.

\begin{figure}
\includegraphics[width=0.5\textwidth]{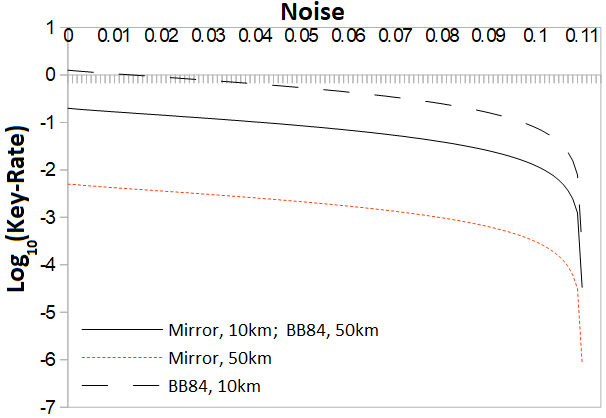}
\caption{A graph of the final key rate
versus the noise level of the Mirror protocol
in the second scenario (single-photon attacks \textit{with} losses),
compared to two copies of BB84,
for two possible lengths of fiber channels
($\ell = 10 \mathrm{km}$ and $\ell = 50 \mathrm{km}$) and
$\alpha = 0.2\frac{\mathrm{dB}}{\mathrm{km}}$.
Note that this figure presents the effective key rate
computed by the expression $r = S(A|E) - H(A|B)$, which scales
with the probability of a raw key bit being generated.
Also note that the key rate of BB84 at $50 \mathrm{km}$ coincides with that
of the Mirror protocol at $10 \mathrm{km}$,
so both are plotted as the same (solid) line.}\label{fig:keyrate-lossy}
\end{figure}

\section{Summary}
We have proved security of the Mirror protocol
against collective attacks, including attacks where
the adversary Eve sends multiple photons towards
the classical user (Alice). Our analysis shows that
the asymptotic noise tolerance of the Mirror protocol is comparable,
in the single-photon scenario, to the ``QKD with Classical Bob''
protocol~\cite{cbob07,QKD-Tom-Krawec-Arbitrary}
and even to the BB84 protocol.
Moreover, we have suggested a general framework for analyzing
multi-photon attacks; this framework may be useful
for other QKD and SQKD protocols, too.

We conclude the Mirror protocol is theoretically
secure against collective attacks, and we suspect
similar security results can be achieved for general attacks.
Extensions of our results, such as security against general attacks,
security against multi-photon attacks on both channels,
and evaluation of our key rate formula in the multi-photon case,
are left for future research.
Our extension to multi-photon attacks also suggests the intriguing
possibility of analyzing SQKD protocols employing decoy states and
similar counter-measures against practical attacks.

Our results show that SQKD protocols can potentially be implemented
in a secure way, overcoming the practical attacks suggested
by~\cite{cbob07comment,cbob07comment_reply}. They therefore hold
the potential to transform the SQKD protocols, making them not only
theoretically fascinating, but also practically secure.

\bibliographystyle{unsrturl}
\bibliography{mirror}

\begin{thebibliography}{10}

\bibitem{mirror17}
Michel Boyer, Matty Katz, Rotem Liss, and Tal Mor.
\newblock Experimentally feasible protocol for semiquantum key distribution.
\newblock {\em Phys. Rev. A}, 96:062335, Dec 2017.
\newblock \href {https://doi.org/10.1103/PhysRevA.96.062335}
  {\path{doi:10.1103/PhysRevA.96.062335}}.

\bibitem{bb84}
Charles~H. Bennett and Gilles Brassard.
\newblock Quantum cryptography: Public key distribution and coin tossing.
\newblock In {\em International Conference on Computers, Systems \& Signal
  Processing, IEEE, 1984}, pages 175--179, Dec 1984.

\bibitem{cbob07}
Michel Boyer, Dan Kenigsberg, and Tal Mor.
\newblock Quantum key distribution with classical {Bob}.
\newblock {\em Phys. Rev. Lett.}, 99:140501, Oct 2007.
\newblock \href {https://doi.org/10.1103/PhysRevLett.99.140501}
  {\path{doi:10.1103/PhysRevLett.99.140501}}.

\bibitem{sqkd09}
Michel Boyer, Ran Gelles, Dan Kenigsberg, and Tal Mor.
\newblock Semiquantum key distribution.
\newblock {\em Phys. Rev. A}, 79:032341, Mar 2009.
\newblock \href {https://doi.org/10.1103/PhysRevA.79.032341}
  {\path{doi:10.1103/PhysRevA.79.032341}}.

\bibitem{calice09}
Xiangfu Zou, Daowen Qiu, Lvzhou Li, Lihua Wu, and Lvjun Li.
\newblock Semiquantum-key distribution using less than four quantum states.
\newblock {\em Phys. Rev. A}, 79:052312, May 2009.
\newblock \href {https://doi.org/10.1103/PhysRevA.79.052312}
  {\path{doi:10.1103/PhysRevA.79.052312}}.

\bibitem{calice09comment}
Michel Boyer and Tal Mor.
\newblock Comment on ``semiquantum-key distribution using less than four
  quantum states''.
\newblock {\em Phys. Rev. A}, 83:046301, Apr 2011.
\newblock \href {https://doi.org/10.1103/PhysRevA.83.046301}
  {\path{doi:10.1103/PhysRevA.83.046301}}.

\bibitem{LC2008}
Hua Lu and Qing-Yu Cai.
\newblock Quantum key distribution with classical {Alice}.
\newblock {\em Int. J. Quantum Inf.}, 06(06):1195--1202, Dec 2008.
\newblock \href {https://doi.org/10.1142/S0219749908004353}
  {\path{doi:10.1142/S0219749908004353}}.

\bibitem{SDL2013}
Zhi-Wei Sun, Rui-Gang Du, and Dong-Yang Long.
\newblock Quantum key distribution with limited classical {Bob}.
\newblock {\em Int. J. Quantum Inf.}, 11(01):1350005, Apr 2013.
\newblock \href {https://doi.org/10.1142/S0219749913500056}
  {\path{doi:10.1142/S0219749913500056}}.

\bibitem{YYLH2014}
Kun-Fei Yu, Chun-Wei Yang, Ci-Hong Liao, and Tzonelih Hwang.
\newblock Authenticated semi-quantum key distribution protocol using {Bell}
  states.
\newblock {\em Quantum Inf. Process.}, 13(6):1457--1465, Mar 2014.
\newblock \href {https://doi.org/10.1007/s11128-014-0740-z}
  {\path{doi:10.1007/s11128-014-0740-z}}.

\bibitem{mediated15}
Walter~O. Krawec.
\newblock Mediated semiquantum key distribution.
\newblock {\em Phys. Rev. A}, 91:032323, Mar 2015.
\newblock \href {https://doi.org/10.1103/PhysRevA.91.032323}
  {\path{doi:10.1103/PhysRevA.91.032323}}.

\bibitem{ZQZM15}
Xiangfu Zou, Daowen Qiu, Shengyu Zhang, and Paulo Mateus.
\newblock Semiquantum key distribution without invoking the classical party's
  measurement capability.
\newblock {\em Quantum Inf. Process.}, 14(8):2981--2996, Aug 2015.
\newblock \href {https://doi.org/10.1007/s11128-015-1015-z}
  {\path{doi:10.1007/s11128-015-1015-z}}.

\bibitem{zhang2020single}
Wei Zhang, Daowen Qiu, and Paulo Mateus.
\newblock Single-state semi-quantum key distribution protocol and its security
  proof.
\newblock {\em Int. J. Quantum Inf.}, 18(04):2050013, Jul 2020.
\newblock \href {https://doi.org/10.1142/S0219749920500136}
  {\path{doi:10.1142/S0219749920500136}}.

\bibitem{rong2021mediated}
Zhenbang Rong, Daowen Qiu, Paulo Mateus, and Xiangfu Zou.
\newblock Mediated semi-quantum secure direct communication.
\newblock {\em Quantum Inf. Process.}, 20(58):1--13, Feb 2021.
\newblock \href {https://doi.org/10.1007/s11128-020-02965-2}
  {\path{doi:10.1007/s11128-020-02965-2}}.

\bibitem{rong2020two}
Zhenbang Rong, Daowen Qiu, and Xiangfu Zou.
\newblock Two single-state semi-quantum secure direct communication protocols
  based on single photons.
\newblock {\em Int. J. Mod. Phys. B}, 34(11):2050106, May 2020.
\newblock \href {https://doi.org/10.1142/S0217979220501064}
  {\path{doi:10.1142/S0217979220501064}}.

\bibitem{rong2020semi}
Zhenbang Rong, Daowen Qiu, and Xiangfu Zou.
\newblock Semi-quantum secure direct communication using entanglement.
\newblock {\em Int. J. Theor. Phys.}, 59:1807--1819, Apr 2020.
\newblock \href {https://doi.org/10.1007/s10773-020-04447-8}
  {\path{doi:10.1007/s10773-020-04447-8}}.

\bibitem{mahadev18}
Urmila Mahadev.
\newblock Classical verification of quantum computations.
\newblock In {\em Proceedings of 2018 IEEE 59th Annual Symposium on Foundations
  of Computer Science (FOCS)}, pages 259--267, Oct 2018.
\newblock \href {https://doi.org/10.1109/FOCS.2018.00033}
  {\path{doi:10.1109/FOCS.2018.00033}}.

\bibitem{reichardt2013classical}
Ben~W. Reichardt, Falk Unger, and Umesh Vazirani.
\newblock Classical command of quantum systems.
\newblock {\em Nature}, 496(7446):456--460, Apr 2013.
\newblock \href {https://doi.org/10.1038/nature12035}
  {\path{doi:10.1038/nature12035}}.

\bibitem{sqc20}
Hasan Iqbal and Walter~O. Krawec.
\newblock Semi-quantum cryptography.
\newblock {\em Quantum Inf. Process.}, 19(3):97, Feb 2020.
\newblock \href {https://doi.org/10.1007/s11128-020-2595-9}
  {\path{doi:10.1007/s11128-020-2595-9}}.

\bibitem{secret-share-1}
Qin Li, W.~H. Chan, and Dong-Yang Long.
\newblock Semiquantum secret sharing using entangled states.
\newblock {\em Phys. Rev. A}, 82:022303, Aug 2010.
\newblock \href {https://doi.org/10.1103/PhysRevA.82.022303}
  {\path{doi:10.1103/PhysRevA.82.022303}}.

\bibitem{secret-share-2}
Lvzhou Li, Daowen Qiu, and Paulo Mateus.
\newblock Quantum secret sharing with classical {Bobs}.
\newblock {\em J. Phys. A}, 46(4):045304, Jan 2013.
\newblock \href {https://doi.org/10.1088/1751-8113/46/4/045304}
  {\path{doi:10.1088/1751-8113/46/4/045304}}.

\bibitem{secret-share-3}
Zhulin Li, Qin Li, Chengdong Liu, Yu~Peng, Wai~Hong Chan, and Lvzhou Li.
\newblock Limited resource semiquantum secret sharing.
\newblock {\em Quantum Inf. Process.}, 17(10):285, Sep 2018.
\newblock \href {https://doi.org/10.1007/s11128-018-2058-8}
  {\path{doi:10.1007/s11128-018-2058-8}}.

\bibitem{semi-sdc-1}
Chen Xie, Lvzhou Li, Haozhen Situ, and Jianhao He.
\newblock Semi-quantum secure direct communication scheme based on {Bell}
  states.
\newblock {\em Int. J. Theor. Phys.}, 57(6):1881--1887, Jun 2018.
\newblock \href {https://doi.org/10.1007/s10773-018-3713-7}
  {\path{doi:10.1007/s10773-018-3713-7}}.

\bibitem{semi-sdc-2}
XiangFu Zou and DaoWen Qiu.
\newblock Three-step semiquantum secure direct communication protocol.
\newblock {\em Sci. China Phys. Mech. Astron.}, 57(9):1696--1702, Sep 2014.
\newblock \href {https://doi.org/10.1007/s11433-014-5542-x}
  {\path{doi:10.1007/s11433-014-5542-x}}.

\bibitem{semi-sdc-3}
Ming-Hui Zhang, Hui-Fang Li, Zhao-Qiang Xia, Xiao-Yi Feng, and Jin-Ye Peng.
\newblock Semiquantum secure direct communication using {EPR} pairs.
\newblock {\em Quantum Inf. Process.}, 16(5):117, Mar 2017.
\newblock \href {https://doi.org/10.1007/s11128-017-1573-3}
  {\path{doi:10.1007/s11128-017-1573-3}}.

\bibitem{semi-sdc-4}
LiLi Yan, YuHua Sun, Yan Chang, ShiBin Zhang, GuoGen Wan, and ZhiWei Sheng.
\newblock Semi-quantum protocol for deterministic secure quantum communication
  using {Bell} states.
\newblock {\em Quantum Inf. Process.}, 17(11):315, Oct 2018.
\newblock \href {https://doi.org/10.1007/s11128-018-2086-4}
  {\path{doi:10.1007/s11128-018-2086-4}}.

\bibitem{semi-ident-1}
Xiao-Jun Wen, Xing-Qiang Zhao, Li-Hua Gong, and Nan-Run Zhou.
\newblock A semi-quantum authentication protocol for message and identity.
\newblock {\em Laser Phys. Lett.}, 16(7):075206, Jun 2019.
\newblock \href {https://doi.org/10.1088/1612-202x/ab232c}
  {\path{doi:10.1088/1612-202x/ab232c}}.

\bibitem{semi-ident-2}
Nan-Run Zhou, Kong-Ni Zhu, Wei Bi, and Li-Hua Gong.
\newblock Semi-quantum identification.
\newblock {\em Quantum Inf. Process.}, 18(6):197, May 2019.
\newblock \href {https://doi.org/10.1007/s11128-019-2308-4}
  {\path{doi:10.1007/s11128-019-2308-4}}.

\bibitem{semi-compare}
Kishore Thapliyal, Rishi~Dutt Sharma, and Anirban Pathak.
\newblock Orthogonal-state-based and semi-quantum protocols for quantum private
  comparison in noisy environment.
\newblock {\em Int. J. Quantum Inf.}, 16(05):1850047, Sep 2018.
\newblock \href {https://doi.org/10.1142/S0219749918500478}
  {\path{doi:10.1142/S0219749918500478}}.

\bibitem{aharonov2017interactive}
Dorit Aharonov, Michael Ben-Or, Elad Eban, and Urmila Mahadev.
\newblock Interactive proofs for quantum computations.
\newblock {\em arXiv preprint arXiv:1704.04487}, Apr 2017.
\newblock URL: \url{https://arxiv.org/abs/1704.04487}.

\bibitem{plug_play97}
A.~Muller, T.~Herzog, B.~Huttner, W.~Tittel, H.~Zbinden, and N.~Gisin.
\newblock ``plug and play'' systems for quantum cryptography.
\newblock {\em Appl. Phys. Lett.}, 70(7):793--795, Feb 1997.
\newblock \href {https://doi.org/10.1063/1.118224}
  {\path{doi:10.1063/1.118224}}.

\bibitem{ping_pong02}
Kim Bostr\"om and Timo Felbinger.
\newblock Deterministic secure direct communication using entanglement.
\newblock {\em Phys. Rev. Lett.}, 89:187902, Oct 2002.
\newblock \href {https://doi.org/10.1103/PhysRevLett.89.187902}
  {\path{doi:10.1103/PhysRevLett.89.187902}}.

\bibitem{LM05}
Marco Lucamarini and Stefano Mancini.
\newblock Secure deterministic communication without entanglement.
\newblock {\em Phys. Rev. Lett.}, 94:140501, Apr 2005.
\newblock \href {https://doi.org/10.1103/PhysRevLett.94.140501}
  {\path{doi:10.1103/PhysRevLett.94.140501}}.

\bibitem{renner-twoway13}
Normand~J. Beaudry, Marco Lucamarini, Stefano Mancini, and Renato Renner.
\newblock Security of two-way quantum key distribution.
\newblock {\em Phys. Rev. A}, 88:062302, Dec 2013.
\newblock \href {https://doi.org/10.1103/PhysRevA.88.062302}
  {\path{doi:10.1103/PhysRevA.88.062302}}.

\bibitem{mirror_gurevich13}
Pavel Gurevich.
\newblock Experimental quantum key distribution with classical {Alice}.
\newblock Master's thesis, Technion---Israel Institute of Technology, Haifa,
  May 2013.
\newblock URL:
  \url{https://www.graduate.technion.ac.il/Theses/Abstracts.asp?Id=26105}.

\bibitem{mirror_tamari14}
Natan Tamari.
\newblock Experimental semiquantum key distribution: Classical {Alice} with
  mirror.
\newblock Master's thesis, Technion---Israel Institute of Technology, Haifa,
  Nov 2014.
\newblock URL:
  \url{https://www.graduate.technion.ac.il/Theses/Abstracts.asp?Id=28660}.

\bibitem{massa2022experimental}
Francesco Massa, Preeti Yadav, Amir Moqanaki, Walter~O. Krawec, Paulo Mateus,
  Nikola Paunkovi{\'{c}}, Andr{\'{e}} Souto, and Philip Walther.
\newblock Experimental semi-quantum key distribution with classical users.
\newblock {\em Quantum}, 6:819, Sep 2022.
\newblock URL: \url{https://quantum-journal.org/papers/q-2022-09-22-819/},
  \href {https://doi.org/10.22331/q-2022-09-22-819}
  {\path{doi:10.22331/q-2022-09-22-819}}.

\bibitem{cbob07comment}
Yong-gang Tan, Hua Lu, and Qing-yu Cai.
\newblock Comment on ``quantum key distribution with classical {Bob}''.
\newblock {\em Phys. Rev. Lett.}, 102:098901, Mar 2009.
\newblock \href {https://doi.org/10.1103/PhysRevLett.102.098901}
  {\path{doi:10.1103/PhysRevLett.102.098901}}.

\bibitem{cbob07comment_reply}
Michel Boyer, Dan Kenigsberg, and Tal Mor.
\newblock {Boyer}, {Kenigsberg}, and {Mor} reply:.
\newblock {\em Phys. Rev. Lett.}, 102:098902, Mar 2009.
\newblock \href {https://doi.org/10.1103/PhysRevLett.102.098902}
  {\path{doi:10.1103/PhysRevLett.102.098902}}.

\bibitem{mirror_attack18}
Michel Boyer, Rotem Liss, and Tal Mor.
\newblock Attacks against a simplified experimentally feasible semiquantum key
  distribution protocol.
\newblock {\em Entropy}, 20(7), Jul 2018.
\newblock \href {https://doi.org/10.3390/e20070536}
  {\path{doi:10.3390/e20070536}}.

\bibitem{cbob_secur15}
Walter~O. Krawec.
\newblock Security proof of a semi-quantum key distribution protocol.
\newblock In {\em 2015 IEEE International Symposium on Information Theory
  (ISIT)}, pages 686--690. IEEE, Jun 2015.
\newblock \href {https://doi.org/10.1109/ISIT.2015.7282542}
  {\path{doi:10.1109/ISIT.2015.7282542}}.

\bibitem{sqkd_secur16}
Walter~O. Krawec.
\newblock Security of a semi-quantum protocol where reflections contribute to
  the secret key.
\newblock {\em Quantum Inf. Process.}, 15(5):2067--2090, Feb 2016.
\newblock \href {https://doi.org/10.1007/s11128-016-1266-3}
  {\path{doi:10.1007/s11128-016-1266-3}}.

\bibitem{calice_secur18}
Wei Zhang, Daowen Qiu, and Paulo Mateus.
\newblock Security of a single-state semi-quantum key distribution protocol.
\newblock {\em Quantum Inf. Process.}, 17(6):135, Apr 2018.
\newblock \href {https://doi.org/10.1007/s11128-018-1904-z}
  {\path{doi:10.1007/s11128-018-1904-z}}.

\bibitem{sqkd_secur18}
Walter~O. Krawec.
\newblock Practical security of semi-quantum key distribution.
\newblock In Eric Donkor, editor, {\em Proceedings of SPIE, Quantum Information
  Science, Sensing, and Computation X}, volume 10660, page 1066009, May 2018.
\newblock \href {https://doi.org/10.1117/12.2303759}
  {\path{doi:10.1117/12.2303759}}.

\bibitem{BM97a}
Eli Biham and Tal Mor.
\newblock Security of quantum cryptography against collective attacks.
\newblock {\em Phys. Rev. Lett.}, 78:2256--2259, Mar 1997.
\newblock \href {https://doi.org/10.1103/PhysRevLett.78.2256}
  {\path{doi:10.1103/PhysRevLett.78.2256}}.

\bibitem{BM97b}
Eli Biham and Tal Mor.
\newblock Bounds on information and the security of quantum cryptography.
\newblock {\em Phys. Rev. Lett.}, 79:4034--4037, Nov 1997.
\newblock \href {https://doi.org/10.1103/PhysRevLett.79.4034}
  {\path{doi:10.1103/PhysRevLett.79.4034}}.

\bibitem{BBBGM02}
Eli Biham, Michel Boyer, Gilles Brassard, Jeroen van~de Graaf, and Tal Mor.
\newblock Security of quantum key distribution against all collective attacks.
\newblock {\em Algorithmica}, 34(4):372--388, Nov 2002.
\newblock \href {https://doi.org/10.1007/s00453-002-0973-6}
  {\path{doi:10.1007/s00453-002-0973-6}}.

\bibitem{mayers01}
Dominic Mayers.
\newblock Unconditional security in quantum cryptography.
\newblock {\em J. ACM}, 48(3):351--406, May 2001.
\newblock \href {https://doi.org/10.1145/382780.382781}
  {\path{doi:10.1145/382780.382781}}.

\bibitem{SP00}
Peter~W. Shor and John Preskill.
\newblock Simple proof of security of the {BB84} quantum key distribution
  protocol.
\newblock {\em Phys. Rev. Lett.}, 85:441--444, Jul 2000.
\newblock \href {https://doi.org/10.1103/PhysRevLett.85.441}
  {\path{doi:10.1103/PhysRevLett.85.441}}.

\bibitem{BBBMR06}
Eli Biham, Michel Boyer, Oscar~P. Boykin, Tal Mor, and Vwani Roychowdhury.
\newblock A proof of the security of quantum key distribution.
\newblock {\em J. Cryptol.}, 19(4):381--439, Apr 2006.
\newblock \href {https://doi.org/10.1007/s00145-005-0011-3}
  {\path{doi:10.1007/s00145-005-0011-3}}.

\bibitem{RGK05}
Renato Renner, Nicolas Gisin, and Barbara Kraus.
\newblock Information-theoretic security proof for quantum-key-distribution
  protocols.
\newblock {\em Phys. Rev. A}, 72:012332, Jul 2005.
\newblock \href {https://doi.org/10.1103/PhysRevA.72.012332}
  {\path{doi:10.1103/PhysRevA.72.012332}}.

\bibitem{renner_thesis08}
Renato Renner.
\newblock Security of quantum key distribution.
\newblock {\em Int. J. Quantum Inf.}, 6(01):1--127, Feb 2008.
\newblock \href {https://doi.org/10.1142/S0219749908003256}
  {\path{doi:10.1142/S0219749908003256}}.

\bibitem{CKR09}
Matthias Christandl, Robert K\"onig, and Renato Renner.
\newblock Postselection technique for quantum channels with applications to
  quantum cryptography.
\newblock {\em Phys. Rev. Lett.}, 102:020504, Jan 2009.
\newblock \href {https://doi.org/10.1103/PhysRevLett.102.020504}
  {\path{doi:10.1103/PhysRevLett.102.020504}}.

\bibitem{geat2022}
Tony Metger and Renato Renner.
\newblock Security of quantum key distribution from generalised entropy
  accumulation.
\newblock {\em arXiv preprint arXiv:2203.04993}, Mar 2022.
\newblock URL: \url{https://arxiv.org/abs/2203.04993}.

\bibitem{guskind2022mediated}
Julia Guskind and Walter~O. Krawec.
\newblock Mediated semi-quantum key distribution with improved efficiency.
\newblock {\em Quantum Sci. Technol.}, 7(3):035019, Jun 2022.
\newblock \href {https://doi.org/10.1088/2058-9565/ac7412}
  {\path{doi:10.1088/2058-9565/ac7412}}.

\bibitem{BLMS00}
Gilles Brassard, Norbert L\"utkenhaus, Tal Mor, and Barry~C. Sanders.
\newblock Limitations on practical quantum cryptography.
\newblock {\em Phys. Rev. Lett.}, 85:1330--1333, Aug 2000.
\newblock \href {https://doi.org/10.1103/PhysRevLett.85.1330}
  {\path{doi:10.1103/PhysRevLett.85.1330}}.

\bibitem{DW05}
Igor Devetak and Andreas Winter.
\newblock Distillation of secret key and entanglement from quantum states.
\newblock {\em Proc. R. Soc. A}, 461(2053):207--235, Jan 2005.
\newblock \href {https://doi.org/10.1098/rspa.2004.1372}
  {\path{doi:10.1098/rspa.2004.1372}}.

\bibitem{QKD-Tom-Krawec-Arbitrary}
Walter~O. Krawec.
\newblock Quantum key distribution with mismatched measurements over arbitrary
  channels.
\newblock {\em Quantum Inf. Comput.}, 17(3 and 4):209--241, Mar 2017.
\newblock \href {https://doi.org/10.26421/QIC17.3-4-2}
  {\path{doi:10.26421/QIC17.3-4-2}}.

\bibitem{beaudry2013security}
Normand~J. Beaudry, Marco Lucamarini, Stefano Mancini, and Renato Renner.
\newblock Security of two-way quantum key distribution.
\newblock {\em Phys. Rev. A}, 88:062302, Dec 2013.
\newblock \href {https://doi.org/10.1103/PhysRevA.88.062302}
  {\path{doi:10.1103/PhysRevA.88.062302}}.

\bibitem{krawec2018key}
Walter~O. Krawec.
\newblock Key-rate bound of a semi-quantum protocol using an entropic
  uncertainty relation.
\newblock In {\em 2018 IEEE International Symposium on Information Theory
  (ISIT)}, pages 2669--2673. IEEE, Jun 2018.
\newblock \href {https://doi.org/10.1109/ISIT.2018.8437303}
  {\path{doi:10.1109/ISIT.2018.8437303}}.

\bibitem{QKD-Tom-First}
Stephen~M. Barnett, Bruno Huttner, and Simon~J.D. Phoenix.
\newblock Eavesdropping strategies and rejected-data protocols in quantum
  cryptography.
\newblock {\em J. Mod. Opt.}, 40(12):2501--2513, 1993.
\newblock \href {https://doi.org/10.1080/09500349314552491}
  {\path{doi:10.1080/09500349314552491}}.

\bibitem{QKD-Tom-KeyRateIncrease}
Shun Watanabe, Ryutaroh Matsumoto, and Tomohiko Uyematsu.
\newblock Tomography increases key rates of quantum-key-distribution protocols.
\newblock {\em Phys. Rev. A}, 78:042316, Oct 2008.
\newblock \href {https://doi.org/10.1103/PhysRevA.78.042316}
  {\path{doi:10.1103/PhysRevA.78.042316}}.

\end{thebibliography}

\EOD
\end{document}